\documentclass[a4paper,12pt]{article}
\usepackage{amsmath, amssymb, amsthm}
\usepackage{graphicx, geometry}
\usepackage{hyperref, todonotes, changepage, caption}
\usepackage{listings}
\usepackage{xcolor}
\usepackage[breakable]{tcolorbox}
\usepackage{fancyvrb}
\usepackage{upquote}
\usepackage{enumitem}
\newtheorem*{remark}{Remark}

\definecolor{cellbackground}{HTML}{F7F7F7}
\definecolor{cellborder}{HTML}{CFCFCF}
\definecolor{incolor}{HTML}{303F9F}
\definecolor{outcolor}{HTML}{D84315}

\theoremstyle{plain}
\newtheorem{theorem}{Theorem}[section]
\newtheorem{lemma}{Lemma}[section]
\theoremstyle{definition}
\newtheorem{definition}{Definition}[section]

\newenvironment{customproof}[1][Proof]{\noindent\textbf{#1.} }{\hfill$\square$\vspace{1em}}

\title{Limit Continuous Poker: A Variant of Continuous Poker with Limited Bet Sizes}
\author{Andrew Spears}
\date{\today}

\begin{document}
\maketitle

\begin{abstract}
We introduce and analyze Limit Continuous Poker, a variant of Von Neumann's Continuous Poker with variable but limited bet sizes. This simplified variant of poker captures aspects of information asymmetry, bluffing, balancing, and the impact of bet size limits while still being simple enough to solve analytically. We derive the Nash equilibrium strategy profile for this game, showing how the bettor's and caller's strategies depend on the bet size limits. We demonstrate that as the bet size limits approach extreme values, the strategy profile converges to those of other continuous poker variants. Finally, we connect these results to strategic implications of limited bet sizing in real-world poker.
\end{abstract}

\tableofcontents
\newpage

\section{Introduction}

Poker is a notoriously complex game, with many different variants and strategic elements that have challenged both players and theorists for decades. Simplified poker models play a crucial role in game theory research by isolating specific strategic aspects---such as bluffing, value betting, and bet sizing---while remaining analytically tractable. One such class of models is Continuous Poker, which abstracts poker hands to continuous numerical hand strengths and restricts play to a single betting round. This simplification allows for exact Nash equilibrium solutions and provides insights that can inform understanding of more complex poker variants.

In this paper, we introduce and analyze Limit Continuous Poker (LCP), a new variant that bridges two well-studied extremes: Fixed-Bet Continuous Poker (FBCP), where the bettor must use a predetermined bet size, and No-Limit Continuous Poker (NLCP), where the bettor can choose any positive bet size. LCP generalizes both by imposing lower and upper bounds $L$ and $U$ on allowable bet sizes, creating a spectrum of games parametrized by these limits.

\subsection{Related Work and Background}

The study of simplified poker models dates back to von Neumann and Morgenstern's foundational work on game theory \cite{vonneumann1944theory}, which introduced the concept of optimal mixed strategies in competitive games. Kuhn \cite{kuhn1950simplified} developed another influential simplified poker model with discrete hands, which has become a standard testbed for game-theoretic algorithms. Continuous poker variants---where hands are drawn from continuous distributions---have since become standard examples in game theory textbooks and research \cite{koller1997representations}, serving as tractable models for studying information asymmetry and strategic bluffing. More recently, computational approaches have achieved remarkable success in larger poker variants, including the solution of heads-up limit Texas hold'em \cite{bowling2015heads}. Our work builds directly on two classical continuous variants: Fixed-Bet Continuous Poker (FBCP) and No-Limit Continuous Poker (NLCP). We briefly review these games to establish context for LCP.

\subsubsection{Fixed-Bet Continuous Poker (FBCP)}

Continuous Poker (also called Von Neumann Poker, and referred to in this paper as Fixed-Bet Continuous Poker or FBCP) is a simplified model of poker introduced by von Neumann. It is a two-player zero-sum game designed to study strategic decision-making in competitive environments. The game abstracts away many complexities of real poker, focusing instead on the mathematical and strategic aspects of bluffing, betting, and optimal play.

\begin{definition}[FBCP]
Two players, referred to as the bettor and the caller, each put a 0.5 unit ante into a pot\footnote{An ante of 1 is often used, but since the pot size is the more relevant value, we use an ante of 0.5. All bet sizes scale proportionally.}. They are each dealt a hand strength between 0 and 1 (referred to as $x$ for bettor and $y$ for caller). After seeing $x$, the bettor can either check --- in which case, the higher hand between $x$ and $y$ wins the pot of 1 and the game ends --- or they can bet, by putting a pre-determined amount $B$ into the pot. The caller must now either call by matching the bet of $B$ units, after which the higher hand wins the pot of $1+2B$ minus their ante of $0.5$, or fold, conceding the pot of $1+B$ to the bettor and ending the game.
\end{definition}

FBCP has many Nash equilibria, but it has a unique one in which the caller plays an admissible strategy\footnote{An admissible strategy is one which is not weakly dominated by any other strategy.}, as shown by Ferguson and Ferguson \cite[p. 2]{ferguson2003borel}. This strategy profile, parametrized by the bet size $B$, is structured as follows:

\begin{itemize}
    \item The bettor bluffs with their weakest hands, below some threshold hand strength
    \item The bettor value bets with their strongest hands, above some threshold hand strength
    \item The bettor checks hands between the two thresholds
    \item The caller calls if their hand is stronger than some threshold hand strength
\end{itemize}






\begin{figure}[h]
    \centering
    \begin{tikzpicture}[scale=4]
        \draw[thick] (0,0) -- (3,0);

        \draw[thick] (0, -0.04) -- (0, 0.04);
        \draw[thick] (3, -0.04) -- (3, 0.04);

        \node[below] at (0, -0.07) {$0$};
        \node[below] at (3, -0.07) {$1$};

        \fill[red!30] (0, -0.04) rectangle (0.6, 0.04);

        \fill[gray!20] (0.6, -0.04) rectangle (2.1, 0.04);

        \fill[red!30] (2.1, -0.04) rectangle (3, 0.04);

        \draw[thick] (0.6, -0.04) -- (0.6, 0.04);
        \draw[thick] (2.1, -0.04) -- (2.1, 0.04);

        \node[below] at (0.6, -0.07) {$\frac{B}{(1+2B)(2+B)}$};
        \node[below] at (2.1, -0.07) {$\frac{1+4B+2B^2}{(1+2B)(2+B)}$};

        \node[red!70!black] at (0.3, 0.14) {\small Bluff};
        \node[gray!70!black] at (1.35, 0.14) {\small Check};
        \node[red!70!black] at (2.55, 0.14) {\small Value bet};

        \node[left] at (-0.05, 0) {\small Bettor:};

        \begin{scope}[yshift=-0.5cm]
            \draw[thick] (0,0) -- (3,0);

            \draw[thick] (0, -0.04) -- (0, 0.04);
            \draw[thick] (3, -0.04) -- (3, 0.04);

            \node[below] at (0, -0.07) {$0$};
            \node[below] at (3, -0.07) {$1$};

            \fill[gray!20] (0, -0.04) rectangle (1.35, 0.04);

            \fill[blue!30] (1.35, -0.04) rectangle (3, 0.04);

            \draw[thick] (1.35, -0.04) -- (1.35, 0.04);

            \node[below] at (1.35, -0.07) {$\frac{B(3+2B)}{(1+2B)(2+B)}$};

            \node[gray!70!black] at (0.675, 0.14) {\small Fold};
            \node[blue!70!black] at (2.175, 0.14) {\small Call};

            \node[left] at (-0.05, 0) {\small Caller:};
        \end{scope}
    \end{tikzpicture}
    \caption{Optimal strategies in FBCP. The bettor partitions hand strengths into bluffing, checking, and value betting regions. The caller folds below a threshold and calls above it.}
    \label{fig:fbcp_strategy}
\end{figure}
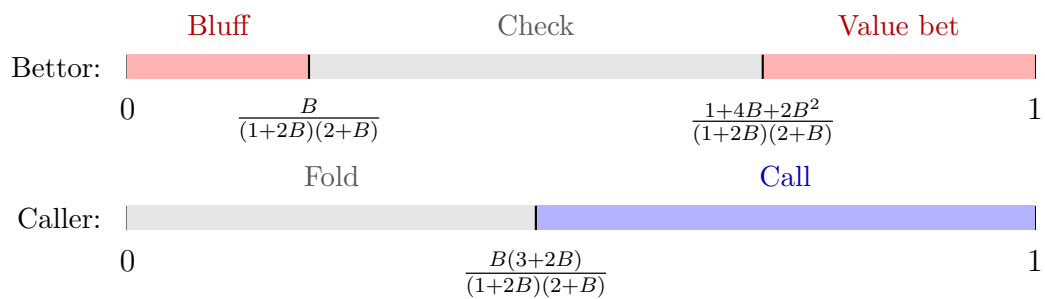

The non-uniqueness of this Nash Equilibrium is due to the fact that given the bettor's strategy, the caller has many optimal responses. The caller must always fold with hands below the bluffing threshold, and must always call with hands above the value betting threshold, but with hands in-between, they are indifferent between calling and folding. This is because with a hand strength in this range, the caller wins if and only if the bettor is bluffing, so their actual hand strength is irrelevant as long as it beats the bluffing threshold. To prevent the betting player from exploiting them, the caller need only call with exactly the right proportion of hands in this range. For example, the caller could take the strategy described above, but swap some calling and folding hands in the range between the bluffing and value betting thresholds. At the bettor's bluffing threshold, they will still be indifferent between bluffing and checking because the same proportion of hands which beat them will still call. Similarly, at the value betting threshold, they still beat all weaker hands, so the order of weaker hands doesn't matter as long as the same proportion is willing to call a bet.

Why is this Nash equilibrium special? We mentioned above that it is admissible, meaning that both players' strategies are not weakly dominated by any other strategy. Importantly, the caller's strategy is not weakly dominated. The same cannot be said for other Nash equilibria like the one described in the previous paragraph because monotonicity (calling with stronger hands and folding with weaker ones) characterizes admissibility in this setup. We explore this theme further in Section \ref{subsec:monotone_strategies}.

FBCP also has a unique value as a function of the bet size $B$. The value of the game for the bettor is positive (advantageous to the bettor) and maximized at $B = 1$, when the bet size is exactly the pot size. It should not be surprising that the value is positive --- at worst, the bettor can always check and turn the game into a coin flip, so the bettor will only deviate from this strategy if they have a positive expected value. The fact that $B=1$ exactly maximizes the value is more subtle, but we should expect that some such maximal value of $B$ exists. Forcing the bettor to bet too large relative to the pot would make betting too risky with most hands, and making the bet too small would simply give less profit to the bettor when they win a bet. As we will see later, part of the motivation for studying LCP is to understand this concept more generally.

\subsubsection{No-Limit Continuous Poker (NLCP)}
Another continuous poker variant allows the bettor to choose a bet size $s > 0$ after seeing their hand strength, as opposed to a fixed bet size $B$. This variant is called No-Limit Continuous Poker (or Newman Poker after Donald J. Newman, or NLCP in this paper). The Nash equilibrium strategy profile for this variant is discussed and solved by Bill Chen and Jerrod Ankenman \cite[p. 154]{chen2006mathematics}.

Nash Equilibrium has the following structure:

\begin{itemize}
    \item The bettor bluffs with their weak hands below some threshold hand strength, according to a decreasing function (weakest hands make the largest bluffs)
    \item The bettor value bets with their strong hands above some threshold hand strength, according to an increasing function (strongest hands make the largest bets)
    \item The bettor checks hands between the two thresholds
    \item The caller has a calling threshold which depends on the bet size. After seeing the bet size, they call with hands above the corresponding threshold and fold those below
\end{itemize}






See Figure \ref{fig:nlcp_strategy_profile} for a graphical representation of the strategy profile.

\begin{figure}[h!]
    \centering
    \includegraphics[width=0.8\textwidth]{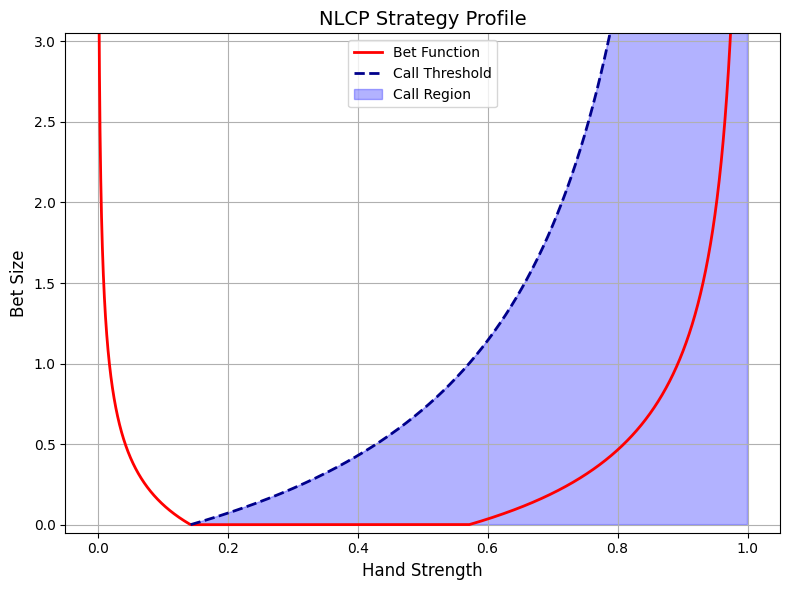}
    \caption{The optimal bettor strategy (red, mapping hand strengths to bet sizes), along with the optimal caller strategy (blue, the region of hand strength-bet size pairs where the caller should call). Note that the x-axis represents either the bettor's hand strength or the caller's, depending on which function is being read.}
    \label{fig:nlcp_strategy_profile}
\end{figure}

Note that the bettor uses all possible bet sizes and has exactly two hand strengths for each positive bet size\footnote{Seen visually in Figure \ref{fig:nlcp_strategy_profile} by the fact that a horizontal line above the x-axis intersects the bet function at exactly two points.}. On first inspection, this feels like the bettor is giving away too much information (seeing their bet size tells the caller that they have one of exactly two possible hands). However, it turns out that this is still be an optimal strategy as long as the two possible hands are `balanced' properly. In fact, the ration of the derivatives of the betting function at these two points turns out to be the crucial factor. This concept appears again and is explained more thoroughly in Section \ref{subsec:nash_equilibrium_structure}.

The value of NLCP is

$$ V_{NL} = \frac{1}{14}, $$

for the bettor\footnote{Would be 1/7 for an ante of 1, but the value is halved with an ante of 0.5.}. Thus, NLCP is again advantageous to the bettor. In fact, one can easily verify that NLCP is more advantageous to the bettor than FBCP for any bet size $B$ by arguing that the bettor could artificially restrict themselves to a single bet size and achieve the same value as the bettor in FBCP.

Ferguson and Ferguson \cite{ferguson2003borel} provided the comprehensive analysis of FBCP described above, establishing the unique admissible Nash equilibrium and deriving closed-form solutions for optimal strategies. Chen and Ankenman \cite{chen2006mathematics} extended this analysis to NLCP, demonstrating how unlimited bet sizing maintains a similar strategic structure but introduces the complexity of strategies as functions with carefully balanced derivatives. Our work builds on these foundations by introducing a parametric family of games that interpolates between these extremes, allowing us to study how betting constraints affect optimal strategies and game value in a continuous fashion.

\subsection{Our Contributions}

This paper makes the following contributions:

\begin{itemize}
    \item \textbf{Nash Equilibrium Solution:} We construct a Nash equilibrium for LCP. We establish the structure of optimal play and provide closed-form expressions for all strategic components (Sections \ref{sec:solving_lcp}--\ref{sec:nash_equilibrium}).

    \item \textbf{Game Value Analysis:} We compute the value of LCP as a function of the betting limits $L$ and $U$, obtaining a surprisingly elegant rational formula. We prove monotonicity properties and establish a remarkable symmetry: $V_{LCP}(L, U) = V_{LCP}(1/U, 1/L)$ (Section \ref{sec:game_value}).
    
    \item \textbf{Convergence Results:} We prove that LCP strategies and values converge to those of FBCP and NLCP in the appropriate limit cases, establishing LCP as a genuine generalization of both variants (Section \ref{sec:strategic_convergence}).

    \item \textbf{Parameter Sensitivity Analysis:} We analyze how changes in betting limits affect optimal strategies and payoffs, revealing counterintuitive effects where expanding betting options can reduce expected value for certain hand strengths due to strategic adjustments by the opponent (Section \ref{sec:parameter_analysis}).
\end{itemize}

\subsection{Outline}

The remainder of this paper is organized as follows. Section 2 establishes notation and conventions. Section \ref{sec:lcp} formally defines LCP. Section \ref{sec:solving_lcp} develops the methodology for solving for Nash equilibrium, including the concepts of monotone and admissible strategies. Section \ref{sec:nash_equilibrium} presents the Nash equilibrium strategy profile with closed-form solutions. Section \ref{sec:game_value} analyzes the game value, proving monotonicity and symmetry properties. Section \ref{sec:strategic_convergence} establishes convergence to FBCP and NLCP. Section \ref{sec:parameter_analysis} examines how parameters affect strategies and payoffs. Section \ref{sec:conclusion} concludes with a discussion of implications and future work. Detailed proofs are provided in the appendices.

\section{Preliminaries and Conventions}

\subsection{Game-Theoretic Concepts}

All continuous poker variants studied in this paper are \textit{two-player zero-sum games}. LCP also has infinite action and type spaces (continuous bet sizes and hand strengths). Unlike finite games, in an infinite game, a Nash equilibrium may or may not exist. For LCP, we explicitly construct a Nash equilibrium in Section \ref{sec:nash_equilibrium} and verify in Appendix \ref{app:nash_equilibrium} that no player can improve by unilateral deviation. This constructive approach establishes existence without appealing to general existence theorems.

Once a Nash equilibrium is established, the game has a well-defined \textit{value}: the expected payoff to a certain player at equilibrium. In two-player zero-sum games, all Nash equilibria yield the same payoff to each player, so this value is unique. Since the game is zero-sum, the caller's expected payoff is the negative of this value.

As is standard in continuous games, we define uniqueness of strategies up to sets of measure zero. For instance, if two strategies take identical actions except at an exact threshold hand strength (which occurs with probability 0), then they are equivalent for our purposes.

\subsection{Conventions}

\begin{description}[style=nextline, leftmargin=1em, font=\bfseries]
    \item[Ante and Pot Size:] Each player contributes an ante of 0.5 units, creating an initial pot of 1 unit. All bet sizes are measured in units relative to this pot. This convention (ante = 0.5 rather than ante = 1) simplifies payoff calculations.
    
    \item[Payoffs:] Payoffs represent the net gain or loss relative to the initial ante. A check results in the winner receiving the pot of 1 minus their ante of 0.5, for a net payoff of $\pm 0.5$. When a bet of size $s$ is called, the winner receives $1 + 2s$ (pot plus both contributions) minus their ante of 0.5 and bet of $s$, for a net payoff of $\pm(0.5 + s)$.
    
    \item[Inequalities and Measure Zero:] When we write, for example, ``the caller calls with hands $y \geq c(s)$,'' the choice of $\geq$ versus $>$ is immaterial because the set $\{y : y = c(s)\}$ has measure zero and hands are uniformly distributed. We use closed intervals $[a,b]$ and open intervals $(a,b)$ interchangeably when the boundary points have probability zero.
    
    \item[Admissibility:] A strategy is \textit{admissible} if it is not weakly dominated by any other strategy.
\end{description}

\section{Limit Continuous Poker (LCP)}
\label{sec:lcp}
We now introduce the subject variant of this paper. In this poker variant, the bettor may choose a bet size $s$ after seeing their hand strength, but $s$ is bounded by an upper limit $U$ and a lower limit $L$, referred to as the maximum and minimum bet sizes. We call this game Limit Continuous Poker (or LCP).

In Limit Continuous Poker, two players compete in a simplified poker game where each receives a hand strength represented by a real number between 0 and 1. The bettor acts first, choosing either to check (pass) or to make a bet of some size within the allowed range. If a bet is made, the caller must decide whether to call the bet or fold.

\begin{definition}[LCP]
A two-player zero-sum game parameterized by $L, U$ satisfying $0 \leq L \leq U$ with gameplay defined by:
\begin{itemize}
    \item The bettor and caller are each dealt independent hand strengths $X, Y \sim \text{Uniform}[0,1]$\footnote{Uniform without loss of generality, since the quantile transform can take any distribution to uniform and only orders of hand strengths are relevant.}
    \item The bettor observes their hand strength $x$ and chooses an action from $\mathcal{A}_1 = \{\text{0}\} \cup [L, U]$ (a bet of 0 is a check)
    \item If the bettor chooses an action from $[L, U]$ (a bet), then the caller observes the bettor's action along with their own hand strength $y$ and chooses from $\mathcal{A}_2 = \{\text{call}, \text{fold}\}$ 
    \item Payoffs are determined as follows:
    \begin{itemize}
        \item If the bettor checks: payoff is $0.5$ (the other player's ante) to the player with higher hand strength
        \item If the bettor bets $s \in [L, U]$ and the caller calls: payoff is $0.5 + s$ (the other player's bet/call and ante) to the player with higher hand strength
        \item If the bettor bets $s \in [L, U]$ and the caller folds: payoff is $0.5$ (the caller's ante) to the bettor
    \end{itemize}
    The losing player's payoff is always the negative of the winning player's payoff, making it a zero-sum game.
\end{itemize}

A strategy for the bettor is a measurable function $\sigma_1: [0,1] \to \mathcal{A}_1$ mapping hand strengths to actions. A strategy for the caller is a measurable function $\sigma_2: [L, U] \times [0, 1] \to \mathcal{A}_2$ mapping bet sizes and caller hand strengths to caller responses.
\end{definition}

The motivation for studying this variant is twofold. First, it is a more realistic variant of poker than FBCP or NLCP, where bets are not fixed but are also not unbounded. In most real variants of poker, bet sizes are constrained by the stack sizes of the players and by a minimum bet size. Analytically solving LCP can give insight into the effect of bet size constraints on more complex variants of poker. Strong poker players have intuition about how bet size constraints affect strategy, but rigorously proving this intuition is often impossible given the combinatorial complexity of the game. Second, LCP can be seen as a generalization of FBCP and NLCP; specifically, as $L \to 0$ and $U \to \infty$, LCP approaches NLCP, and as $L \to B$ and $U \to B$ for some fixed value $B$, LCP approaches FBCP. Studying LCP can help us understand the relationship between these two and answer questions about why they produce the strategies they do (see Section \ref{sec:strategic_convergence} for formal convergence results).

In the next section, we develop the methodology for constructing a Nash equilibrium of LCP. Section \ref{subsec:nash_equilibrium_structure} will describe the structure of this equilibrium, and the complete closed-form solution is presented in Section \ref{sec:nash_equilibrium}.

\section{Solving for Nash Equilibrium}
\label{sec:solving_lcp}

In this section, we develop the methodology for computing a Nash equilibrium of LCP. While LCP has many Nash equilibria, we choose to study a specific one with several nice properties. Our approach proceeds in three stages: (1) establishing desired properties of the equilibrium, (2) characterizing the structure that such a Nash equilibrium must satisfy, and (3) deriving a system of equations whose solution yields the equilibrium strategy profile. The complete derivation and verification that our solution constitutes a Nash equilibrium is provided in Appendix \ref{app:nash_equilibrium}.

\subsection{Equilibrium Selection}

Like FBCP, LCP has an infinite class of Nash equilibria. For the purposes of this paper, we construct a specific Nash equilibrium which align with intuition and interpretability. Specifically, we find a Nash equilibrium where strategies can be represented as continuous functions where possible (e.g. mapping hand strengths to bet sizes). Second, we look for a \textit{monotone} calling strategy (defined below). Finally, we choose a betting strategy which orders bluffs in weakly decreasing size, motivated by a less rigorous but intuitive argument. These refinements are not strictly necessary, but give us an equilibrium which aligns well with traditional poker intuition and is easy to describe mathematically.

\subsection{Monotone Strategies}
\label{subsec:monotone_strategies}

\begin{definition}[Monotone Calling Strategy]
    A \textit{monotone} calling strategy is a pure strategy such that for any bet size $s$ and any two hand strengths $y_1 < y_2$, if the caller calls a bet of size $s$ with $y_1$, they must also call with $y_2$.
\end{definition}

This should sound intuitive. Violating this condition (in a non-negligible way) is actually weakly dominated. Not only is a monotone strategy weakly better against all opponents, but there exists an opponent against which the monotone strategy is strictly better (see Appendix \ref{sec:monotone_proofs} for proof). Restricting to pure strategies can be explained similarly: it is better to always call with a stronger hand and always fold a weaker one than to mix between the two.

This gives us a way to narrow the space of possible calling strategies. But what about betting strategies? We mentioned earlier that many optimal betting strategies differ only in how they bluff. The bettor always bluffs with their weakest hands. They will have some threshold hand strength such that every hand below this will bluff, but the sizes of these bluffs and ordering of the sizes within this hand strength region are unprescribed. When the caller plays optimally, the caller never calls with a hand below this bluffing threshold, so they never lose to a bluff. However, if the caller deviates to a suboptimal strategy, then the bettor's bluffing hand strength might come into play. We need to make a certain bluff size a certain proportion of the time to balance our value bets, but how should we order these bluffs?

Conventional poker wisdom dictates that the bettor should make their largest bluffs with their weakest hands, and their smaller bluffs with the stronger of the (still weak) bluffing hands\cite{chen2006mathematics}. It turns out that there is a mathematical justification for this intuition, but that it makes many assumptions about the caller's strategy. In traditional poker, we generally expect that the caller will either call with the optimal frequency, `overcall' (call with weaker hands than they should), or `undercall' (fold hands they should call with). If the caller overcalls, and if they are more likely to call smaller bets (another strong assumption), then the bettor might `accidentally' win a showdown when bluffing. However, this is far more likely to happen if the bettor makes their small bluffs with their strongest bluffing hands. If they make their small bluffs with their strongest hands, then they never win these overcalling situations. We will not try to formalize this argument because it relies on too many restrictions on the caller's strategy, but it does provide justification for why this sort of bluffing strategy is desirable. 

Thus, we choose to analyze a Nash equilibrium in which the bettor orders their bluffs in non-increasing size as a function of hand strength (largest bluffs with weakest hands). Keep in mind that we will formally prove that the constructed strategy profile is a Nash equilibrium--this discussion is purely motivation for the specific choice of equilibrium.

\subsection{Nash Equilibrium Structure}
\label{subsec:nash_equilibrium_structure}

We will now describe the structure of the Nash equilibrium in terms of thresholds $x_i$ and functions $c(s)$, $b(s)$, and $v(s)$. These turn out to be fully determined by the parameters $L$ and $U$, but for now they are unknown. Notice that both players use pure strategies, like in NLCP: the bettor maps hand strengths directly to bet sizes, and the caller maps hand strengths and bet sizes to actions with no mixing. We can break strategies into piecewise regions as follows (see figure \ref{fig:strategyprofile} for a visual representation):

    \begin{itemize}
        \item The caller has a calling threshold $c(s)$ that is continuous and differentiable in $s$, including at endpoints $L$ and $U$. They call with hands $y \geq c(s)$ and fold with hands $y < c(s)$\footnote{The action taken at the threshold is irrelevant, since it occurs with probability zero.}.
        \item The bettor partitions $[0,1]$ into three regions: bluffing $x \in [0,x_2]$, checking $x \in [x_2,x_3]$, and value betting $x \in [x_3,1]$.
        \item Within the bluffing region, the bettor partitions into a max-betting region $x \in [0,x_0]$, an intermediate region $x \in [x_0,x_1]$, and a min-betting region $x \in [x_1,x_2]$ (this is using the assumption of non-increasing bluffing sizes mentioned above).
        \item Within the intermediate bluffing region, the bettor bets according to a continuous, decreasing function $s=b^{-1}(x)$ with endpoints $b^{-1}(x_0)=U$ and $b^{-1}(x_1)=L$. Note that $b(s)$ gives hand strength as a function of bet size, which feels backwards, but turns out to be more mathematically convenient and still equally descriptive.
        \item Within the value betting region, the bettor partitions into a min-betting region $x \in [x_3,x_4]$, an intermediate region $x \in [x_4,x_5]$, and a max-betting region $x \in [x_5,1]$.
        \item Within the intermediate value betting region, the bettor bets according to a continuous, increasing function $s=v^{-1}(x)$ with endpoints $v^{-1}(x_4)=L$ and $v^{-1}(x_5)=U$. Once again, $v(s)$ gives hand strength as a function of bet size, which feels backwards but is more mathematically convenient.
    \end{itemize}

\subsection{Constraints and Indifference Equations}
\label{subsec:constraints}

Having established the qualitative structure of the Nash equilibrium, we now derive the quantitative relationships that the strategy profile must satisfy. These constraints arise from two fundamental equilibrium conditions: players must be indifferent among actions at exact thresholds, and players must optimize when choosing among available actions. The resulting system of differential and algebraic equations will uniquely determine all threshold values $x_0, \ldots, x_5$ and the functions $b(s)$, $v(s)$, and $c(s)$, which in turn uniquely determine the strategy profile.

The key conditions are:

\begin{itemize}
    \item The caller must be indifferent between calling and folding at their calling threshold
    \item The bettor must be indifferent between checking and betting at their value betting and bluffing thresholds
    \item The bettor's bet size for a value bet must maximize their expected value
    \item The bettor's strategy must be continuous in bet size (in the regions where they bet)
\end{itemize}

These conditions give us the following system of equations, which we will derive in the next section: \label{eq:nash_equilibrium_system}

\textbf{Caller Indifference:}
\begin{align}
    & (x_2-x_1) \cdot (1+L) - (x_4-x_3) \cdot L = 0 \\
    & x_0 \cdot (1+U) - (1-x_5) \cdot U = 0\\
    & |b'(s)| \cdot (1 + s) - |v'(s)| \cdot s = 0
\end{align}

\textbf{Bettor Indifference and Optimality:}
\begin{align}
    & -sc'(s) - c(s) + 2 v(s) - 1 = 0 \label{eq:valueoptimality}\\
    & (x_3-c(L)) \cdot (1+L) - (1-x_3) \cdot (L) + c(L) = x_3 \label{eq:valueindiff}\\ 
    & c(s) - (1-c(s)) \cdot s = x_2 \label{eq:bluffindiff}
\end{align}

\textbf{Continuity Constraints:}
\begin{align}
    & b(U) = x_0 \\
    & b(L) = x_1 \\
    & v(U) = x_5 \\
    & v(L) = x_4.
\end{align}

 Note that for this analysis, it is simpler to pretend that payoffs exclude the initial ante of 0.5, since this is a sunk cost to both players and optimization only depends on relative payoffs between actions.

\subsubsection{Caller Indifference}
\label{subsec:caller_indifference}

By definition, $c(s)$ is the threshold above which the caller calls and below which they fold. This means that in Nash Equilibrium, the caller must be indifferent between calling and folding with a hand strength of $c(s)$:

  \[  \mathbb{E}[\text{call } c(s)] = \mathbb{E}[\text{fold } c(s)] \]
  \[  \mathbb{P}[\text{bluff} | s] \cdot (1+s) - \mathbb{P}[\text{value bet} | s]\cdot s = 0. \]

We now split into cases based on the value of $s$.

\textbf{Case 1: $s = L$}. The hands the bettor value bets $L$ with are $x \in (x_3, x_4)$, and the hands they bluff with are $x \in (x_1, x_2)$. 

\begin{equation}{\label{callindiffmin}}
    (x_2-x_1) \cdot (1+L) - (x_4-x_3) \cdot L = 0.
\end{equation}

Here, we are implicitly multiplying both sides by the common denominator of $(x_4-x_3) + (x_2-x_1)$.

\textbf{Case 2: $s = U$}. The hands the bettor value bets $U$ with are $x \in (x_5, 1)$, and the hands they bluff with are $x \in (0, x_0)$. 

\begin{equation}{\label{callindiffmax}}
    (1-x_5) \cdot (1+U) - x_0 \cdot U = 0,
\end{equation}

again, implicitly multiplying both sides by the common denominator of $(1-x_5) + x_0$.

\textbf{Case 3: $L \leq s \leq U$}. In this case, the bettor has exactly one value hand and one bluffing hand, but somewhat paradoxically, they are not equally likely. The probability of a value bet given the size $s$ is related to the inverse derivative of the value function $v(s)$ at $s$, and the same goes for a bluff. This gives us the following relation:

\[ \frac{\mathbb{P}[\text{value bet} | s]}{\mathbb{P}[\text{bluff} | s]} = \frac{|b'(s)|}{|v'(s)|}\]

An intuitive interpretation of this is that for any small neighborhood around the bet size $s$, the bettor has more hands which use a bet size in the neighborhood if $v(s)$ does not change rapidly around $s$, that is, if $|v'(s)|$ is small. The same goes for bluffing hands, and as we limit the neighborhood to a single point, the ratio of the two probabilities approaches the ratio of the derivatives. We know that these are the only two possible bettor actions for such a bet size, so

\[ \mathbb{P}[\text{value bet} | s] = \frac{|b'(s)|}{|b'(s)| + |v'(s)|} \]
\[ \mathbb{P}[\text{bluff} | s] = \frac{|v'(s)|}{|b'(s)| + |v'(s)|} \]

Plugging this into the indifference equation and dividing out the common denominator, we get:

\begin{equation}{\label{callindiff}}
    |b'(s)| \cdot (1 + s) - |v'(s)| \cdot s = 0.
\end{equation}

\subsubsection{Bettor Indifference and Optimality}

When the bettor makes a value bet, they are attempting to maximize the expected value of the bet. We can write the expected value of a value bet as:

\begin{align*}
    \mathbb{E}[\text{value bet } s | x] & = \mathbb{P}[\text{call with worse}] \cdot (1+s) - \mathbb{P}[\text{call with better}] \cdot s + \mathbb{P}[\text{fold}] \cdot 1 \\
    & = (x-c(s)) \cdot (1+s) - (1-x) \cdot (s) + c(s).
\end{align*}

To maximize this, we take the derivative with respect to $s$ and set it equal to zero. Crucially, we are treating $c(s)$ as a function of $s$ and using the chain rule, since changing the bet size $s$ will also change the calling threshold $c(s)$. We want this optimality condition to hold for the bettor's Nash equilibrium strategy, so we set $x=v(s)$. This gives us:

\begin{align}{\label{valueoptimality}}
    \nonumber \frac{d}{ds} \mathbb{E}[\text{value bet } s | x=v(s)] & = 0 \\
    -sc'(s) - c(s) + 2 v(s) - 1 & = 0.
\end{align}

Additionally, when the bettor has the most marginal value betting hand at $x=x_3$, they should be indifferent between a minimum value bet and a check: 

\begin{align}{\label{valueindiff}}
    \nonumber \mathbb{E}[\text{value bet } L | x=x_3] & = \mathbb{E}[\text{check} | x=x_3]\\ 
    (x_3-c(L)) \cdot (1+L) - (1-x_3) \cdot (L) + c(L) & = x_3.
\end{align}

Finally, when the bettor has the most marginal bluffing hand at $x=x_2$, they should be indifferent between a minimum bluff and a check. However, as we discussed earlier, the bettor should be indifferent among all bluffing sizes, so the bettor should actually be indifferent between checking and making any bluffing size $s$ at $x=x_2$. This gives us:

\begin{align}{\label{bluffindiff}}
    \nonumber \mathbb{E}[\text{bluff } s | x=x_2] & = \mathbb{E}[\text{check} | x=x_2]\\ 
    c(s) - (1-c(s)) \cdot s & = x_2.
\end{align}

\subsubsection{Continuity Constraints}

As discussed above, the bettor's strategy is continuous in $s$ and $x$ (except when checking). This means that the endpoints of the functions $v(s)$ and $b(s)$ are constrained as follows:

\begin{equation}{\label{continuityconstraints}}
	 b(U) = x_0, \;\; b(L) = x_1, \;\; v(U) = x_5, \;\; v(L) = x_4.
\end{equation}

\section{Nash Equilibrium Strategy Profile}
\label{sec:nash_equilibrium}
Having established the equilibrium structure in Section \ref{sec:solving_lcp} and derived the indifference equations in Section \ref{subsec:constraints}, we now present the complete solution. Figure \ref{fig:strategyprofile} displays the Nash equilibrium strategy profiles for various betting limits, ranging from lenient ($L=0, U=10$) to restrictive ($L=0.5, U=1$).

\begin{figure}
    \begin{adjustwidth}{-1in}{-1in}
        \centering
        \begin{minipage}{0.6\textwidth}
            \centering
            \includegraphics[width=\textwidth]{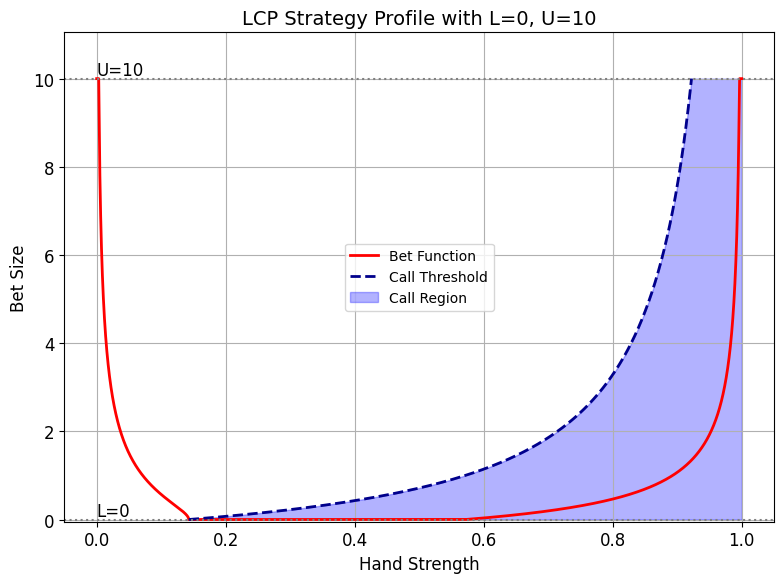}
        \end{minipage}
        \hspace{0.05\textwidth}
        \begin{minipage}{0.6\textwidth}
            \centering
            \includegraphics[width=\textwidth]{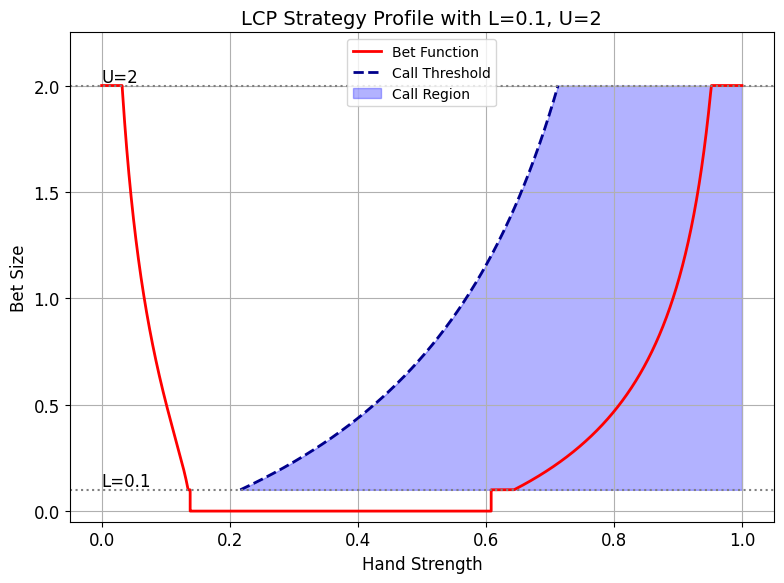}
        \end{minipage}
        \vspace{0.5cm}\\
        \begin{minipage}{0.6\textwidth}
            \centering
            \includegraphics[width=\textwidth]{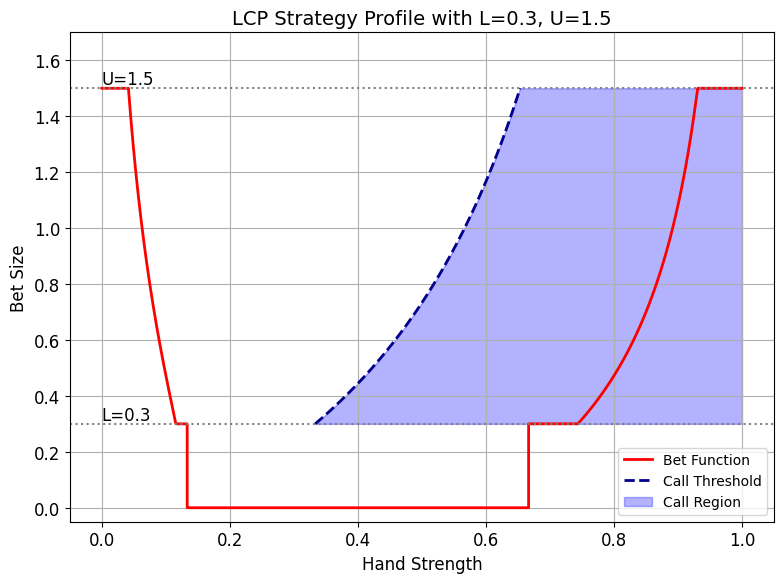}
        \end{minipage}
        \hspace{0.05\textwidth}
        \begin{minipage}{0.6\textwidth}
            \centering
            \includegraphics[width=\textwidth]{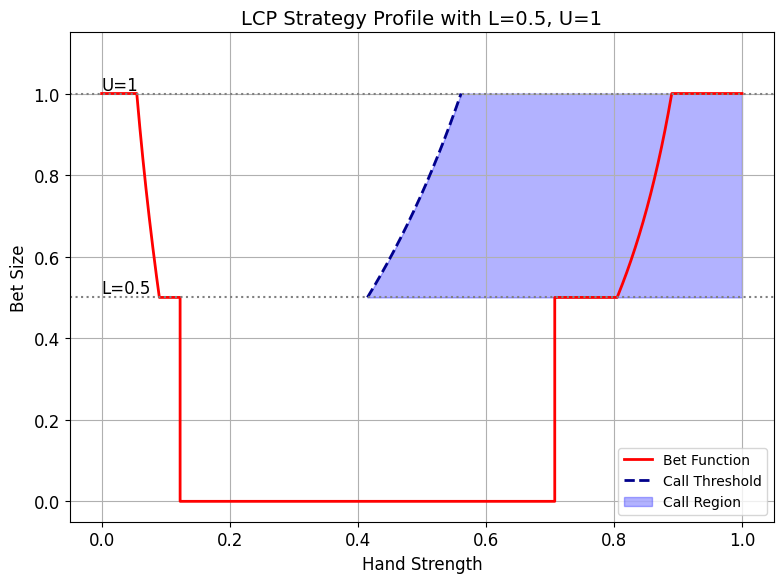}
        \end{minipage}
    \end{adjustwidth}
    \caption{Nash equilibrium strategy profiles for different values of $L$ and $U$, from lenient to restrictive bet size limits. The bet function maps hand strengths to bet sizes, while the call function gives the minimum calling hand strength for a given bet size. The shaded regions represent the hands for which the caller should call.}
    \label{fig:strategyprofile}
\end{figure}

The system of equations from Section \ref{subsec:constraints} was solved symbolically using SymPy (see \ref{app:nash_equilibrium}), yielding closed-form expressions for all threshold values and strategic functions.

\pagebreak
\begin{theorem}[LCP Nash Equilibrium]
    \label{thm:nash_equilibrium}
The following strategy profile constitutes a Nash equilibrium of LCP:

\begin{align*}
    x_{0} &= \frac{3 t^{2} \left(t - 1\right)}{r^{3} + t^{3} - 7}\\
    x_{1} &= \frac{- 2 r^{3} + 3 r^{2} + t^{3} - 1}{r^{3} + t^{3} - 7}\\
    x_{2} &= \frac{r^{3} + t^{3} - 1}{r^{3} + t^{3} - 7}\\
    x_{3} &= \frac{r^{3} - 3 r + t^{3} - 4}{r^{3} + t^{3} - 7}\\
    x_{4} &= \frac{r^{3} + 3 r^{2} - 6 r + t^{3} - 4}{r^{3} + t^{3} - 7}\\
    x_{5} &= \frac{r^{3} + t^{3} + 3 t^{2} - 7}{r^{3} + t^{3} - 7}\\
    b_{0} &= \frac{t^{3}}{r^{3} + t^{3} - 7}\\
    b(s) &= \frac{t^{3} \left(s+1\right)^3 - (3s + 1)}{\left(r^{3} + t^{3} - 7\right) \left(s+1\right)^3}\\
    c(s) &= \frac{r^{3} + t^3 -1 + s \left(r^{3} + t^{3} - 7\right)}{\left(s + 1\right) \left(r^{3} + t^{3} - 7\right)}\\
    v(s) &= \frac{r^{3} + t^{3} -1 + \left(r^{3} + t^{3} - 7\right) \left(2 s^{2} + 4 s + 1\right)}{2 \left(r^{3} + t^{3} - 7\right) \left(s^{2} + 2 s + 1\right)}
\end{align*}

where $r = L/(1+L)$ and $t = 1/(1+U)$.

Where these threshold values and functions describe the complete strategy profile as outlined in Section \ref{subsec:nash_equilibrium_structure}.

\end{theorem}

A proof that this profile constitutes a Nash equilibrium is provided in Appendix \ref{app:nash_equilibrium}.

It is worth noting the change of variables to $(r, t)$ which significantly simplifies the expressions compared to the original $(L, U)$ formulation. This transformation reveals underlying symmetries and makes many properties more transparent, as we will see in section \ref{sec:game_value}.

It is worth noticing that the strategy profile visually approaches NLCP as $L \to 0$ and $U \to \infty$, and approaches FBCP as $L, U \to B$ for any fixed bet size $B$. We formalize these convergence results in Section \ref{sec:strategic_convergence}.

\section{Game Value}
\label{sec:game_value}

Having established that the strategy profile in Section \ref{sec:nash_equilibrium} constitutes a Nash equilibrium (verified in Appendix \ref{app:nash_equilibrium}), the game has a well-defined value: the expected payoff when both players employ these strategies. In two-player zero-sum games, the existence of a Nash equilibrium implies the existence of a unique game value, since all Nash equilibria yield the same expected payoff to each player. This value measures how favorable the game is to the bettor under the given betting limits $L$ and $U$.

\begin{theorem}
    \label{thm:game_value}
    The value of Limit Continuous Poker is given by:
    \[
        V_{LCP}(L, U) = \frac{(1 + L)^3(1+U)^3 - ((1+L)^3+L^3(1 + U)^3)}{14(1 + L)^3(1+U)^3 - 2((1+L)^3+L^3(1 + U)^3)}
    \]
    Equivalently, using the change of variables $r = L/(1+L)$ and $t = 1/(1+U)$, this can be written more compactly as:
    \[
        V(r, t) = \frac{1 - (r^3 + t^3)}{14 - 2(r^3 + t^3)}
    \]
\end{theorem}

Proving this is straightforward in principle but computationally intensive. The value of the game is the expected payoff of these strategies, averaged over all hand pairs $(x, y)$. This reduces to an integral over the unit square (see Figure \ref{fig:payoffs} in Section 7). 

This integral is nontrivial because we are integrating a function over many piecewise regions, and because the bet size is only defined implicitly in terms of the hand strength $x$. The computation requires breaking the unit square into regions based on the strategies and bet sizes, then computing the expected value as a weighted sum of payoffs over these regions. This was again done symbolically using SymPy (see Appendix \ref{app:nash_equilibrium}), yielding the closed-form expression above.

Given this computational complexity, the formula itself is surprisingly simple. The change of variables to $r$ and $t$ reveals an even more elegant structure, along with the following symmetry property:

\begin{align*}
    V_{LCP}(L, U) &= V_{LCP}\left(\frac{1}{U}, \frac{1}{L}\right)\\
    \text{or equivalently, } V(r, t) &= V(t, r)
\end{align*}

The symmetry becomes immediately apparent in the $(r,t)$ formulation: the formula could be parameterized by the unified quantity $r^3 + t^3$. This is not at all obvious from the game setup in terms of $L$ and $U$, even in hindsight.

Figure \ref{fig:game_value_fig} shows the game value as a function of $L, U$ and $r, t$, making the symmetry $V(r,t) = V(t,r)$ visually apparent as a reflection across the diagonal. The diagonal of the left plot ($r+t=1$, or equivalently $L=U$) represents the boundary where the game reduces to fixed-bet continuous poker.

In terms of the original parameters, the symmetry $V_{LCP}(L, U) = V_{LCP}(1/U, 1/L)$ tells us that the benefit from increasing $U$ is exactly equivalent to that of decreasing $L$ in a reciprocal manner, centered around the pot size of $1$. For example, suppose you are given the choice between playing LCP as the bettor with limits $L=1/2$ and $U=5$ or $L=1/5$ and $U=2$. You know you can play optimally, but it is unclear which game favors you. The symmetry property tells us that the value of the game is the same in both cases, so you should be indifferent between the two.

\begin{figure}[h!]
    \centering
    \includegraphics[width=1.1\textwidth]{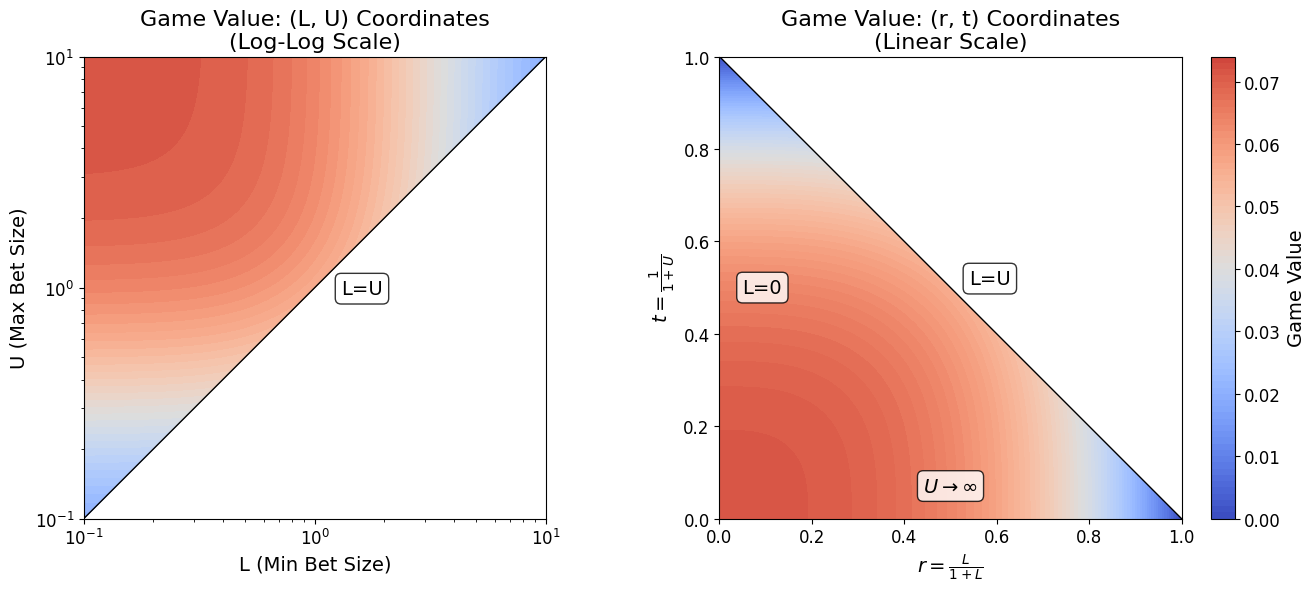}
    \caption{Game value as a function of both parametrizations. The symmetry about the diagonal is immediately visible, corresponding to $V(r,t) = V(t,r)$ or $V(L, U) = V(1/U, 1/L)$. Note that the left plot is cutting off extremely large and small values of $L$ and $U$, while the right plot shows every possible parameter combination.}
    \label{fig:game_value_fig}
\end{figure}

\subsection{Interpreting the Parameters r and t}

Before analyzing the properties of the game value, it is helpful to understand what the transformed parameters $r$ and $t$ represent in game-theoretic terms.

\textbf{The parameter $r = L/(1+L)$ as minimum pot odds:}
When the bettor makes a minimum bet of size $L$, the pot grows from 1 to $1+L$. The caller must risk $L$ to call and potentially win a pot of $1+L$. Thus, $r = L/(1+L)$ represents the \emph{pot odds} the caller receives when facing a minimum bet---the ratio of what they risk to the total pot. This is the most favorable pot odds any caller ever faces in LCP, since larger bets offer worse pot odds. As $L \to 0$, we have $r \to 0$, meaning the minimum bet becomes negligible and calling becomes essentially free. As $L \to \infty$, we have $r \to 1$, meaning the minimum bet becomes prohibitively expensive relative to the pot.

\textbf{The parameter $t = 1/(1+U)$ as pot fraction at maximum bet:}
When the bettor makes a maximum bet of size $U$, the pot grows from 1 to $1+U$. The parameter $t = 1/(1+U)$ represents the original pot as a fraction of the total pot after a maximum bet. Equivalently, $1-t = U/(1+U)$ represents the pot odds the caller receives when facing a maximum bet. A small value of $t$ (close to 0) indicates that $U$ is very large relative to the pot, allowing the bettor to make very aggressive bets. As $U \to \infty$, we have $t \to 0$, meaning the maximum bet becomes arbitrarily large. As $U \to 0$, we have $t \to 1$, meaning the maximum bet becomes negligible.

\textbf{Duality:}
The parameter $r$ fundamentally controls the \emph{caller's incentive to call} with marginal hands: higher $r$ means the minimum bet offers worse pot odds, discouraging calls. The parameter $t$ fundamentally controls the \emph{bettor's ability to apply pressure}: lower $t$ means the maximum bet can be much larger relative to the pot, allowing more aggressive play. The symmetry $V(r,t) = V(t,r)$ reveals a deep duality in the game: swapping the ``minimum calling incentive'' (measured by $r$) with the ``ability to apply pressure'' (measured inversely by $t$) produces games with identical value. This is remarkable because $r$ is primarily about the caller's decisions (pot odds when facing small bets) while $t$ is primarily about the bettor's decisions (how large they can bet), yet these two forces are perfectly balanced in determining the game's value.

In the following sections, we investigate the properties and behavior of $V_{LCP}(L, U)$ in more detail. These include monotonicity, convergence to NLCP and FBCP, and the symmetry property.

\subsection{Value Monotonicity}

Intuitively, more options for the bettor should increase the game's value. Notice the higher value for more lenient limits (red regions of Figure \ref{fig:game_value_fig}) and lower value for more strict limits (blue corners). We can easily prove this formally:

\begin{theorem}
    The value of Limit Continuous Poker is weakly monotonically increasing in $U$ and weakly monotonically decreasing in $L$:
\[
    \frac{\partial V_{LCP}(L, U)}{\partial U} \geq 0, \;\; \frac{\partial V_{LCP}(L, U)}{\partial L} \leq 0.
\]
\end{theorem}
\begin{customproof}
    We can express the derivatives in terms of the cleaner $(r,t)$ variables. Since $r = L/(1+L)$ and $t = 1/(1+U)$, we have:
    \begin{align*}
        \frac{dr}{dL} = \frac{1}{(1+L)^2}, \quad \frac{dt}{dU} = -\frac{1}{(1+U)^2}
    \end{align*}

    Using the chain rule and the fact that $V(r,t) = \frac{1-r^3-t^3}{14-2r^3-2t^3}$:
    \begin{align*}
        \frac{\partial V}{\partial r} &= \frac{-3r^2(14-2r^3-2t^3) + 2 \cdot 3r^2(1-r^3-t^3)}{(14-2r^3-2t^3)^2} = \frac{-18r^2(2-r^3-t^3)}{(14-2r^3-2t^3)^2} < 0\\
        \frac{\partial V}{\partial t} &= \frac{-18t^2(2-r^3-t^3)}{(14-2r^3-2t^3)^2} < 0
    \end{align*}

    where the inequalities hold since $r, t \in (0,1)$ implies $r^3 + t^3 < 2$ and $14 - 2r^3 - 2t^3 > 0$. Therefore:
    \begin{align*}
        \frac{\partial V_{LCP}}{\partial L} &= \frac{\partial V}{\partial r} \frac{dr}{dL} = \frac{-18r^2(2-r^3-t^3)}{(14-2r^3-2t^3)^2} \cdot \frac{1}{(1+L)^2} < 0 \\
        \frac{\partial V_{LCP}}{\partial U} &= \frac{\partial V}{\partial t} \frac{dt}{dU} = \frac{-18t^2(2-r^3-t^3)}{(14-2r^3-2t^3)^2} \cdot \left(-\frac{1}{(1+U)^2}\right) > 0
    \end{align*}
\end{customproof}

\subsection{Value Convergence}

Consistent with the strategic convergence results in Section \ref{sec:strategic_convergence}, the game value also converges to FBCP and NLCP at the appropriate limits. Since $V_{LCP}(L,U)$ is a rational function, these limits follow by direct substitution.

\begin{remark}[Value Convergence to FBCP]
For any $B > 0$, 

\[ \lim_{L,U \to B} V_{LCP}(L, U) = V_{FB}(B) = \frac{B}{2(1+2B)(2+B)} \].

This corresponds to the main diagonal of Figure \ref{fig:game_value_fig}. Note that the maximum on this diagonal occurs at $B=1$, consistent with the known result that a pot-sized bet maximizes the bettor's advantage in FBCP.
\end{remark}

\begin{remark}[Value Convergence to NLCP]
\[ \lim_{L \to 0, U \to \infty} V_{LCP}(L, U) = V_{NL} = \frac{1}{14} \]. 

This corresponds to the top-left corner of the $(L,U)$ plot in Figure \ref{fig:game_value_fig}, or equivalently the origin of the $(r,t)$ plot.
\end{remark}

\section{Strategic Comparison to Fixed-Bet and No-Limit Continuous Poker}
\label{sec:strategic_comparison}

LCP interpolates between Fixed-Bet Continuous Poker (FBCP) and No-Limit Continuous Poker (NLCP). We now make this precise by showing that the LCP strategies converge to those of FBCP and NLCP as the betting limits approach appropriate boundary values.

\subsection{Notation}

We use the following notation for strategy functions across variants:

\begin{itemize}
    \item $S_{FB}(x, B)$, $C_{FB}(s, B)$: Bettor and caller strategies in FBCP with fixed bet $B$
    \item $S_{NL}(x)$, $C_{NL}(s)$: Bettor and caller strategies in NLCP
    \item $S_{LCP}(x, L, U)$, $C_{LCP}(s, L, U)$: Bettor and caller strategies in LCP
\end{itemize}

The FBCP strategies are:
\begin{align*}
	S_{FB}(x, B) & = \begin{cases}
    B & x < \frac{B}{(1+2B)(2+B)} \text{ (bluff)}\\
    0 & \frac{B}{(1+2B)(2+B)} < x < \frac{1 + 4B + 2B^2}{(1+2B)(2+B)} \text{ (check)}\\
    B & x > \frac{1 + 4B + 2B^2}{(1+2B)(2+B)} \text{ (value)}
\end{cases}\\
C_{FB}(B) & = \frac{B(3 +2B)}{(1+2B)(2+B)}.
\end{align*}

The NLCP strategies are:
\begin{align*}
    S_{NL}(x) &= \begin{cases}
        b_{NL}^{-1}(x) & x < \frac{1}{7} \text{ (bluff)} \\
        0 & \frac{1}{7} < x < \frac{4}{7} \text{ (check)} \\
        v_{NL}^{-1}(x) & x > \frac{4}{7} \text{ (value)}
    \end{cases} \\
    C_{NL}(s) &= 1 - \frac{6}{7 (s+1)},
\end{align*}
where $v_{NL}(s) = 1 - \frac{3}{7(s+1)^2}$ and $b_{NL}(s) = \frac{3s+1}{7(s+1)^3}$.

\subsection{Strategic Convergence}
\label{sec:strategic_convergence}

The LCP strategy functions are rational in $L$ and $U$, so convergence follows from continuity by direct substitution of limit values.

\begin{remark}[Convergence to FBCP]
As $L, U \to B$ for any fixed $B > 0$:
\[
S_{LCP}(x, L, U) \to S_{FB}(x, B), \quad C_{LCP}(s, L, U) \to C_{FB}(s, B).
\]
At $L = U = B$, the intermediate bet-size regions collapse ($x_0 = x_1$ and $x_4 = x_5$), leaving only the fixed bet $B$ as an option. The thresholds reduce to:
\[
x_2 = \frac{B}{(1+2B)(2+B)}, \quad x_3 = \frac{2B^2+4B+1}{(1+2B)(2+B)},
\]
which match the FBCP bluff/check and check/value boundaries.
\end{remark}

\begin{remark}[Convergence to NLCP]
As $L \to 0$ and $U \to \infty$:
\[
S_{LCP}(x, L, U) \to S_{NL}(x), \quad C_{LCP}(s, L, U) \to C_{NL}(s).
\]
At these limits, the boundary regions collapse ($x_0 \to 0$, $x_5 \to 1$), and the thresholds become:
\[
x_1 = x_2 \to \frac{1}{7}, \quad x_3 = x_4 \to \frac{4}{7},
\]
matching the NLCP bluff/check and check/value boundaries. The bet-size functions converge to
\[
b(s) \to \frac{3s+1}{7(s+1)^3}, \quad v(s) \to 1 - \frac{3}{7(s+1)^2},
\]
which are exactly $b_{NL}(s)$ and $v_{NL}(s)$. The calling threshold converges to $c(s) \to 1 - \frac{6}{7(1+s)} = C_{NL}(s)$.
\end{remark}

These convergence results confirm that Nash equilibria of LCP smooth interpolate between optimal strategies in the two existing variants.

\section{Parameter and Payoff Analysis}

Having established the Nash equilibrium (Section 4), analyzed the game value (Section \ref{sec:game_value}), and proved convergence to FBCP and NLCP (Section \ref{sec:strategic_comparison}), we now explore in greater detail how the parameters $L$ and $U$ affect player strategies and payoffs. This section summarizes key insights and presents visualizations that illuminate the strategic dynamics of LCP. Complete technical proofs are provided in Appendices \ref{sec:parameter_analysis} and \ref{sec:payoff_analysis}.

\subsection{Visualizing Payoffs in Equilibrium}

In Nash equilibrium, each hand combination $(x, y)$ uniquely determines the bettor's payoff. Figure \ref{fig:payoffs} shows how these payoffs vary across the unit square for different values of $L$ and $U$, from strict limits (Fixed-Bet, $L=U=1$) to lenient limits (No-Limit, $U \to \infty$).

\clearpage
\newgeometry{top=0in,left=0.5in,right=0.5in,bottom=0.5in}
\begin{figure}[p]
    \begin{adjustwidth}{-1.25in}{-1.25in}
        \centering
        \begin{minipage}{0.4\textwidth}
            \centering
            \includegraphics[width=\textwidth]{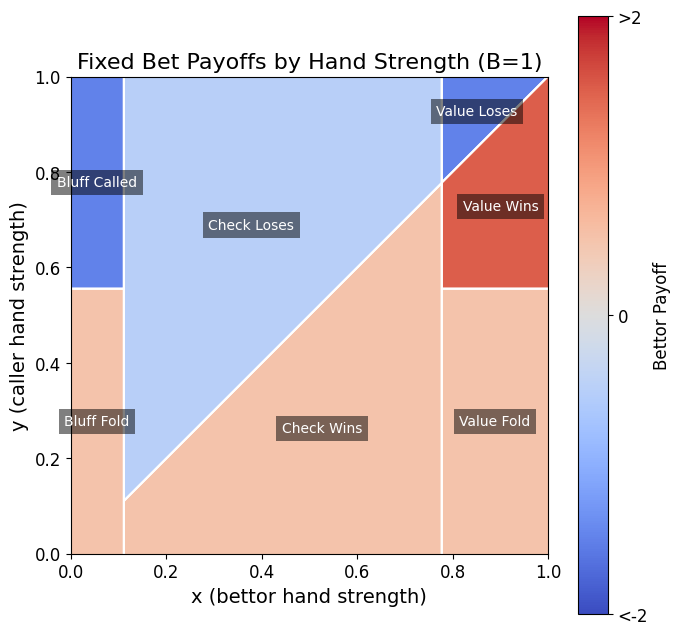}
        \end{minipage}
        \hspace{0.05\textwidth}
        \begin{minipage}{0.4\textwidth}
            \centering
            \includegraphics[width=\textwidth]{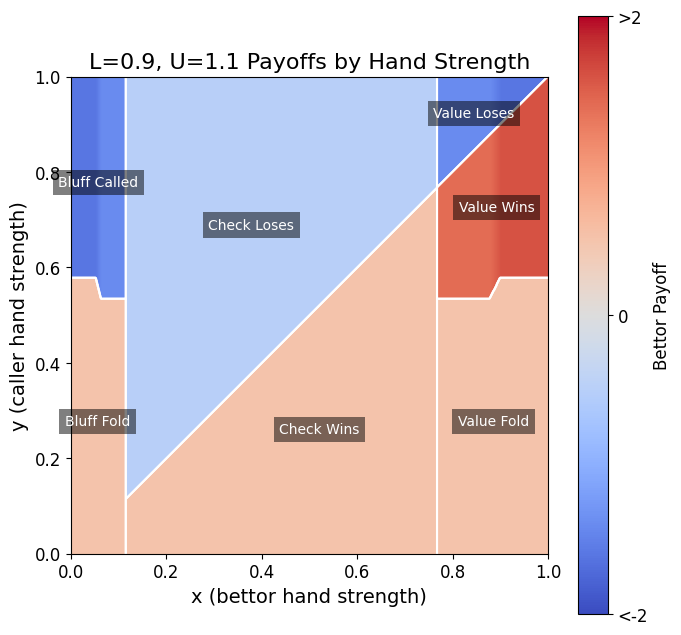}
        \end{minipage}
        \vspace{0.4cm}\\
        \begin{minipage}{0.4\textwidth}
            \centering
            \includegraphics[width=\textwidth]{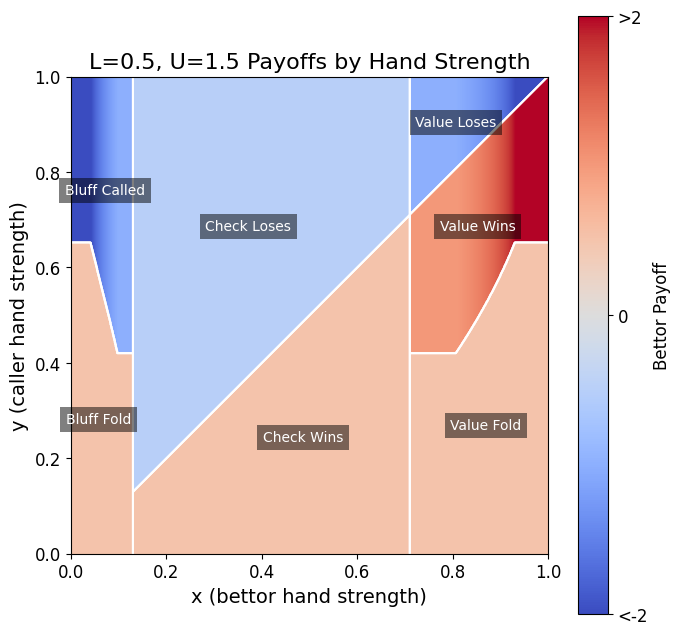}
        \end{minipage}
        \hspace{0.05\textwidth}
        \begin{minipage}{0.4\textwidth}
            \centering
            \includegraphics[width=\textwidth]{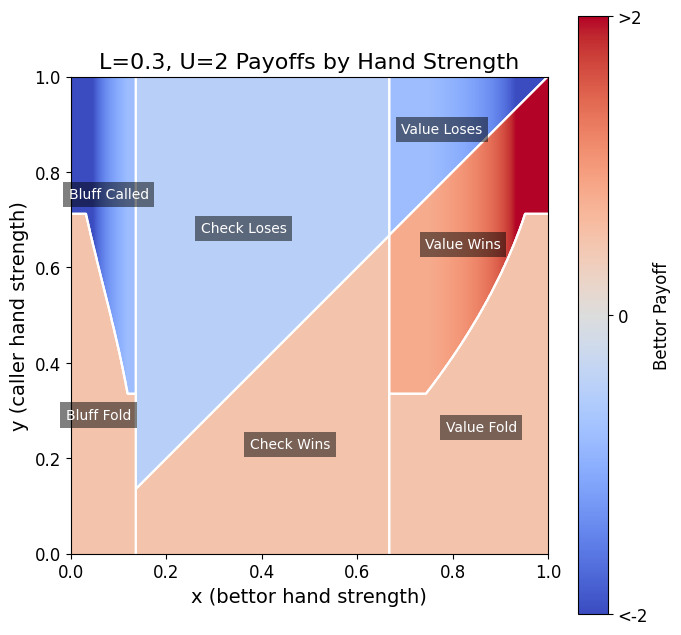}
        \end{minipage}
        \vspace{0.4cm}\\
        \begin{minipage}{0.4\textwidth}
            \centering
            \includegraphics[width=\textwidth]{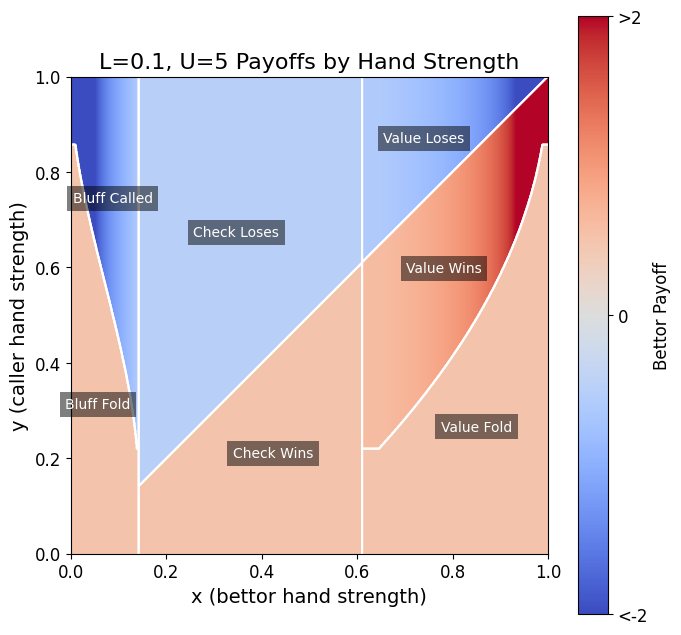}
        \end{minipage}
        \hspace{0.05\textwidth}
        \begin{minipage}{0.4\textwidth}
            \centering
            \includegraphics[width=\textwidth]{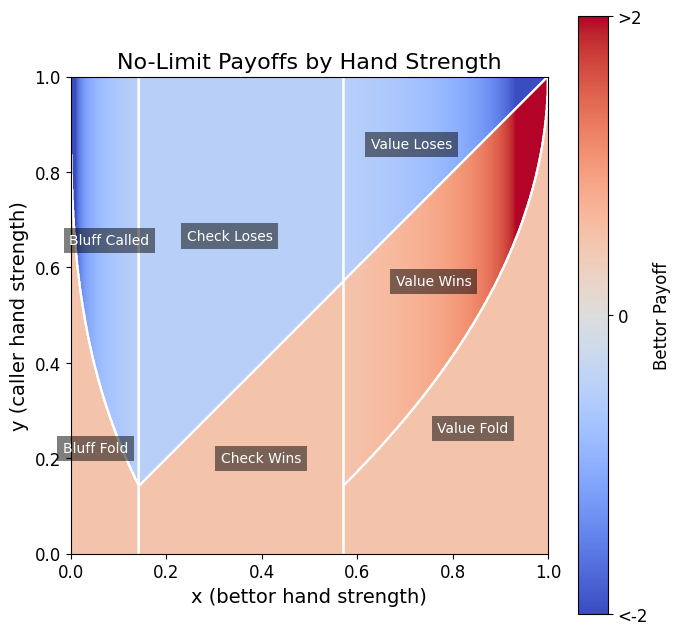}
        \end{minipage}
    \end{adjustwidth}
    \caption{Bettor payoffs in Nash equilibrium as a function of hand strengths $x, y$ for fixed bet size $B=1$ (top left), and No-Limit Continuous Poker (bottom right). Intermediate plots show the payoffs for different values of $L$ and $U$ ranging from strict (fixed bet size $B=1$) to lenient (No limits). Regions are labeled according to the outcome of the game in Nash equilibrium.}
    \label{fig:payoffs}
\end{figure}

\restoregeometry

The visualization reveals that the biggest wins and losses occur when both hands are strong (top right), consistent with real poker intuition. Large payoffs also occur when a very weak bettor bluffs big and gets called by a strong caller (top left). As limits become more lenient, these extreme outcomes become more pronounced but also less likely, since making and calling maximum bets become riskier for both players.

\subsection{Expected Value by Hand Strength}

Beyond specific hand matchups, we can analyze the expected value $EV(x)$ of a bettor hand $x$ averaged over all possible caller hands. This function characterizes how profitable each hand is in equilibrium.

\begin{figure}[h!]
    \centering
    \includegraphics[width=\textwidth]{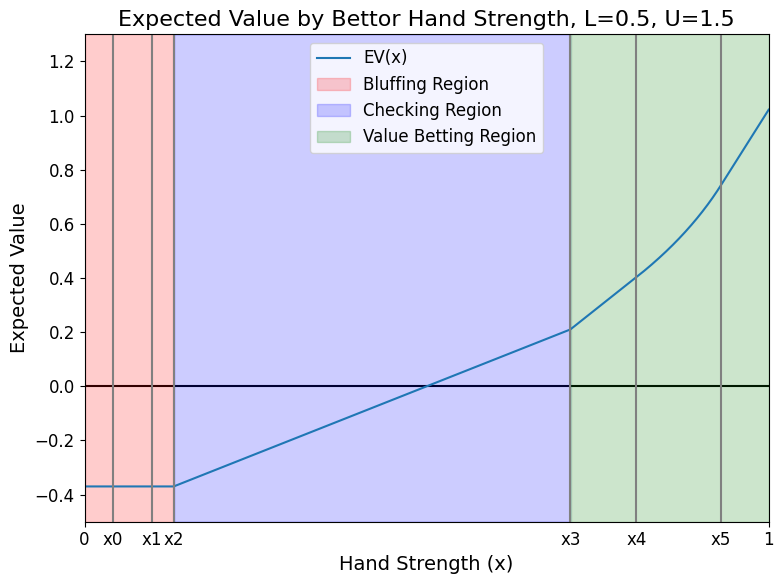}
    \caption{Expected value of bettor hand strength $x$ in the Nash equilibrium from section \ref{sec:nash_equilibrium} for a certain choice of parameters.}
    \label{fig:ev_x}
\end{figure}

Key observations:
\begin{itemize}
    \item All bluffing hands ($x \leq x_2$) achieve the same expected value $x_2 - 1/2$, regardless of hand strength
    \item Checking hands ($x_2 < x \leq x_3$) earn $x - 1/2$ (the ante, won when the bettor has the best hand)
    \item Value betting hands ($x > x_3$) earn increasing returns, with the strongest hands making large bets that win huge pots when called
    \item The function $EV(x)$ is increasing in $x$ (Appendix \ref{sec:payoff_analysis}, Theorem \ref{thm:ev_increasing})
\end{itemize}

\subsection{Effect of Increasing the Upper Limit $U$}

A counterintuitive result emerges when examining how individual hand values change as we increase $U$: for most hand strengths, the expected value \emph{decreases} beyond a certain threshold of $U$ (see Figure \ref{fig:ev_x_vs_U}). This occurs despite the bettor having strictly more strategic options.

\begin{figure}[h!]
    \centering
    \includegraphics[width=\textwidth]{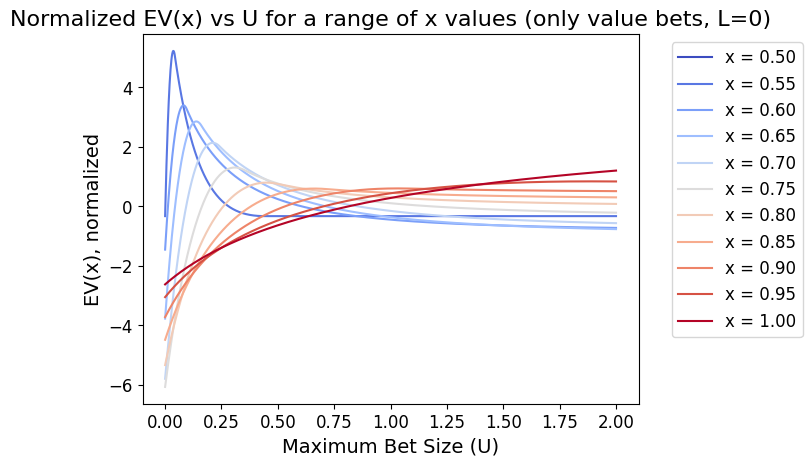}
    \caption{Expected value of a value-betting hand $x$ versus the upper limit $U$ in Nash equilibrium. Each curve increases in $U$ up to some maximum, after which it decreases. This counterintuitive phenomenon is explained by strategic adjustments: as $U$ increases, the caller becomes more conservative (Appendix \ref{sec:parameter_analysis}).}
    \label{fig:ev_x_vs_U}
\end{figure}

The explanation lies in strategic interdependence: as $U$ increases, the bettor can make larger bets with their strongest hands, which forces the caller to become more conservative across \emph{all} bet sizes. This defensive adjustment by the caller harms the expected value of intermediate-strength hands, even though they're betting less. Only the very strongest hands (above a threshold greater than $v(U)$) benefit from the increased flexibility.

The complete analysis in Appendix \ref{sec:parameter_analysis} proves:
\begin{itemize}
    \item The bluffing threshold $x_2$ increases with $U$ (more hands bluff)
    \item For fixed hand $x \in [x_3, v(U)]$, the bet size $v^{-1}(x)$ decreases with $U$
    \item The calling cutoff $c(v^{-1}(x))$ increases with $U$ despite smaller bets
    \item There exists a threshold hand strength above which $EV(x)$ increases with $U$, and below which it decreases
\end{itemize}

These results demonstrate the rich strategic dynamics of LCP, where expanding betting options creates complex ripple effects throughout the equilibrium strategy profile.

\section{Conclusion}
\label{sec:conclusion}
We have introduced and analyzed Limit Continuous Poker (LCP), a parametric family of simplified poker games that bridges the gap between Fixed-Bet Continuous Poker (FBCP) and No-Limit Continuous Poker (NLCP). By imposing lower and upper bounds $L$ and $U$ on bet sizes, LCP creates a rich spectrum of strategic environments that interpolate continuously between the fixed-bet and no-limit extremes.

\subsection{Summary of Key Results}

Our analysis has yielded several main contributions:

\textbf{Nash Equilibrium Characterization:} We constructed a Nash equilibrium for LCP with desirable properties, providing closed-form expressions for all strategic components. Since LCP is a two-player zero-sum game, the existence of this equilibrium establishes a unique game value, and the equilibrium strategies are security strategies for both players. The equilibrium exhibits an elegant structure where the bettor partitions hands into bluffing, checking, and value betting regions, with bet sizes varying continuously within the bluffing and value betting ranges. The caller responds with a calling threshold that depends on bet size, creating a delicate balance where both players are indifferent among their equilibrium actions.

\textbf{Game Value Formula:} We computed the value of LCP as a rational function of the betting limits, obtaining the surprisingly compact expression
\[
V(r,t) = \frac{1 - (r^3 + t^3)}{14 - 2(r^3 + t^3)}
\]
in the transformed coordinates $r = L/(1+L)$ and $t = 1/(1+U)$. This formula reveals a remarkable symmetry: $V(r,t) = V(t,r)$, meaning that swapping the roles of minimum and maximum bet constraints (in a specific reciprocal sense) leaves the game value unchanged. We proved that the value is monotonically increasing in $U$ and decreasing in $L$, confirming the intuition that more betting flexibility favors the bettor.

\textbf{Convergence to Limiting Cases:} We established that LCP smoothly converges to both FBCP and NLCP in the appropriate limit regimes. As $L \to B$ and $U \to B$, the strategies and value converge to those of FBCP with fixed bet size $B$. As $L \to 0$ and $U \to \infty$, they converge to those of NLCP. These results validate LCP as a genuine generalization of both classical variants.

\textbf{Parameter Sensitivity and Strategic Dynamics:} Our analysis revealed counterintuitive strategic effects: increasing the upper limit $U$ does not uniformly benefit all bettor hands. While the strongest hands gain from the ability to make larger bets, intermediate-strength hands can suffer because the caller adjusts by becoming more conservative across all bet sizes. This illustrates the complex strategic interdependencies in equilibrium play.

\subsection{Strategic Insights and Connections to Real Poker}

The theoretical results for LCP offer several insights relevant to practical poker strategy:

\textbf{Bet Sizing and Stack Depth:} In real poker, effective stack sizes create implicit upper bounds on bet sizes, analogous to our parameter $U$. Our analysis suggests that deeper stacks (higher $U$) create more strategic complexity and favor skilled players who can exploit the additional betting options. The symmetry property $V(L,U) = V(1/U, 1/L)$ suggests a duality between raising minimum bet requirements and constraining maximum bets.

\textbf{Bluffing Frequency and Bet Size:} The equilibrium strategy confirms poker wisdom that bluffing frequency should be calibrated to bet size---larger bluffs require fewer bluffing combinations to remain balanced. Moreover, our constructed equilibrium prescribes bluffing larger with weaker hands and smaller with stronger bluffs, which aligns with the practical consideration that stronger bluffs have showdown value.

\textbf{Calling Thresholds and Pot Odds:} The calling function $c(s)$ embodies the pot odds principle: the caller must be offered better odds (a larger pot relative to the cost of calling) to justify calling with weaker hands. The equilibrium precisely balances the bettor's bluffing and value betting frequencies against these pot odds.

\subsection{Limitations}

Several simplifications distinguish LCP from real poker:

\begin{itemize}
    \item \textbf{Single Betting Round:} Real poker involves multiple streets of betting with community cards revealed between rounds, creating dynamic information revelation. LCP models only a single decision point.

    \item \textbf{Uniform Hand Distribution:} We assume hand strengths are uniformly distributed on $[0,1]$. Real poker hand distributions are discrete and non-uniform, with specific card combinations determining hand strength.

    \item \textbf{Perfect Correlation:} In our model, hands are perfectly ordered (if $x > y$, the bettor always wins). Real poker has card removal effects and some hands have non-zero equity even when behind.

    \item \textbf{Symmetric Information:} Both players receive one hand each. Real poker often features asymmetric information structures (e.g., one player knows the other folded a certain range on an earlier street).
\end{itemize}

Despite these limitations, the tractability of LCP enables rigorous analysis that would be impossible in more complex settings, providing a foundation for understanding strategic principles.

\subsection{Future Directions}

Several natural extensions of this work could deepen our understanding of bet sizing in poker:

\textbf{Multiple Betting Rounds:} Extending LCP to multiple streets with information revelation between rounds would capture dynamic aspects of poker strategy. How do bet size limits in early rounds affect optimal play in later rounds?

\textbf{Asymmetric Limits:} Our model assumes both players face the same ante and pot size. Investigating games where players have different effective stack sizes (asymmetric $U$ values) could model scenarios common in tournament poker.

\textbf{Discrete Approximations:} Real poker involves discrete bet sizing increments (e.g., betting in whole chips or minimum raise increments). Studying discrete approximations to LCP could bridge the gap between our continuous model and practical applications.

\textbf{Computational Tools:} The closed-form solutions for LCP could be used to validate numerical solvers for more complex poker variants. The smooth parameter dependence makes LCP an ideal test bed for computational game theory algorithms.

\textbf{Multi-Player Extensions:} While our analysis focuses on two-player games, extending to three or more players would introduce new strategic considerations such as collusion, side pots, and positional dynamics.

The analytical tractability of Limit Continuous Poker makes it a valuable model system for exploring fundamental questions about betting, bluffing, and strategic bet sizing. We hope this work inspires further research into the rich interplay between game structure and optimal strategy in poker and related competitive decision-making environments.

\pagebreak

\appendix

\section{Monotone Strategy Proofs}
\label{sec:monotone_proofs}

This appendix provides the detailed proofs regarding monotone calling strategies and their role in equilibrium selection in LCP.

\subsection{Monotone Strategies and Weak Dominance}

\begin{lemma}
    \label{lem:monotone_dominated}
    If a calling strategy violates the first monotonicity condition for a nonzero-measure set of hands for any bet size $s$, it is weakly dominated. Specifically, if there exists $s$ and measurable sets $A, B \subseteq [0, 1]$ such that:
    \begin{enumerate}
        \item The caller calls $s$ with hands in $A$
        \item The caller folds $s$ with hands in $B$
        \item $\sup A \leq \inf B$
        \item $A$ and $B$ have positive measure
    \end{enumerate}
    then the strategy is weakly dominated.
\end{lemma}

\begin{proof}
    Let $\sigma_C$ be the non-monotone strategy described above. Since $A$ and $B$ are nonzero-measure, there exist subsets $A' \subseteq A$ and $B' \subseteq B$ such that:
    \begin{enumerate}
        \item $A'$ and $B'$ have positive measure
        \item $|A'| = |B'|$
    \end{enumerate}
    Where $|A|$ and $|B|$ denote the Lebesgue measure of $A$ and $B$ (see Figure \ref{fig:strategy_improvement}).

    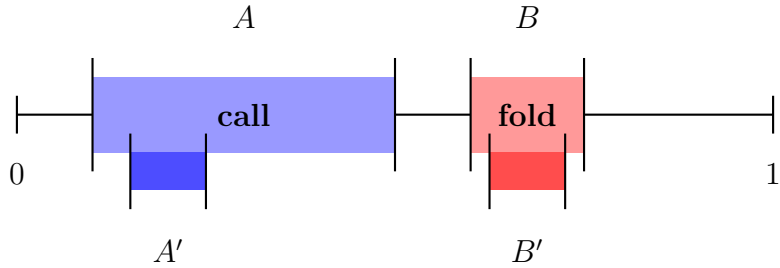
\begin{figure}[h]
        \centering
        \begin{tikzpicture}[scale=5]
            \draw[thick] (0,0) -- (2,0);

            \draw[thick] (0, -0.05) -- (0, 0.05);
            \draw[thick] (2, -0.05) -- (2, 0.05);

            \node[below] at (0, -0.1) {$0$};
            \node[below] at (2, -0.1) {$1$};

            \fill[blue!40] (0.2, -0.1) rectangle (1.0, 0.1);
            \draw[thick] (0.2, -0.15) -- (0.2, 0.15);
            \draw[thick] (1.0, -0.15) -- (1.0, 0.15);
            \node[above] at (0.6, 0.2) {$A$};
            \node at (0.6, 0) {\textbf{call}};

            \fill[red!40] (1.2, -0.1) rectangle (1.5, 0.1);
            \draw[thick] (1.2, -0.15) -- (1.2, 0.15);
            \draw[thick] (1.5, -0.15) -- (1.5, 0.15);
            \node[above] at (1.35, 0.2) {$B$};
            \node at (1.35, 0) {\textbf{fold}};

            \fill[blue!70] (0.3, -0.2) rectangle (0.5, -0.1);
            \draw[thick] (0.3, -0.25) -- (0.3, -0.05);
            \draw[thick] (0.5, -0.25) -- (0.5, -0.05);
            \node[below] at (0.4, -0.3) {$A'$};

            \fill[red!70] (1.25, -0.2) rectangle (1.45, -0.1);
            \draw[thick] (1.25, -0.25) -- (1.25, -0.05);
            \draw[thick] (1.45, -0.25) -- (1.45, -0.05);
            \node[below] at (1.35, -0.3) {$B'$};

        \end{tikzpicture}
        \caption{A simple case of sets $A$ and $B$ which violate monotonicity ($\sup A \leq \inf B$). We can find equal-measure subsets $A' \subseteq A$ and $B' \subseteq B$ to swap actions, improving the strategy.}
        \label{fig:strategy_improvement}
    \end{figure}

    The existence of such subsets follows from a fundamental property of nonatomic measures: since the uniform distribution on $[0,1]$ is nonatomic (no single point has positive probability), for any two measurable sets with positive measure, we can always find measurable subsets of equal measure \cite{Sierpinski1922}. This property allows us to construct the strategy improvement described below.

    Let $\sigma_C'$ be the strategy which switches the actions for $A'$ and $B'$, i.e. calls with $B'$ and folds with $A'$ (and behaves identically for all other bet sizes). We now analyze how this change affects the caller's performance against any betting strategy.

    Against a bet of size $s$, the key improvement occurs in two scenarios:
    \begin{enumerate}
        \item When $y \in B'$ and $x \in A'$: $\sigma_C$ folds while $\sigma_C'$ calls and wins (since $x \in A$ and $y \in B$ with $\sup A \leq \inf B$)
        \item When $y \in A'$ and $x \in B'$: $\sigma_C$ calls and loses while $\sigma_C'$ folds (avoiding the loss)
    \end{enumerate}

    For all other cases, $\sigma_C$ and $\sigma_C'$ behave identically, so $\sigma_C'$ is weakly better than $\sigma_C$ against every betting strategy.

    To show that $\sigma_C$ is strictly dominated, consider a betting strategy which always bets $s$. Against this strategy, both scenarios above occur with positive probability (since $A'$ and $B'$ have positive measure), so $\sigma_C'$ is strictly better than $\sigma_C$. Thus, $\sigma_C$ is weakly dominated.
\end{proof}

\section{Appendix: Proof of Nash Equilibrium}
\label{app:nash_equilibrium}

We now show formally that the strategy profile described in Section \ref{sec:nash_equilibrium} is truly a Nash Equilibrium.

\begin{customproof}
    To show that this is a Nash equilibrium, we need to show that no player can improve their payoff by unilaterally deviating from the strategy profile.
    
    In the proof, we assume that all of the constraints outlined in the previous section are satisfied. The solution was obtained by solving for $c(s)$ in terms of $x_2$, then using this to solve for $v(s)$, and finally solving for $b(s)$ up to a constant of integration. The resulting system of 7 equations in 7 unknowns was solved symbolically using Sympy. The full python script is available \href{https://github.com/andrew-spears/poker_variant_analysis/blob/main/notebooks/limit_continuous_poker/solve_sympy_final.ipynb}{on GitHub}.

    \begin{enumerate}
        \item \textbf{Caller's Deviation:}
            Fix the bet size $s$ and consider the caller's payoff from either calling or folding for each hand strength $y$. 

            \begin{align*}
                \mathbb{E}[\text{call} | y, s] &= \mathbb{P}[x < y | s](1+s) + \mathbb{P}[x \geq y | s](-s) \\
                \mathbb{E}[\text{fold} | y, s] &= 0 \\
            \end{align*}
            
            From section \ref{subsec:caller_indifference}, we know that the expected value of a call is exactly 0 for $y = c(s)$ (by design). The expected value of calling is weakly increasing in $y$, so it must be weakly greater than 0 for $y > c(s)$ and weakly less than 0 for $y < c(s)$. 

            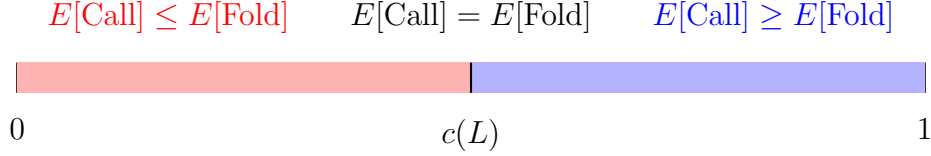
\begin{figure}[h]
                \centering
                \begin{tikzpicture}[scale=4]
                    \draw[thick] (0,0) -- (3,0);
                    
                    \draw[thick] (0, -0.05) -- (0, 0.05);
                    \draw[thick] (3, -0.05) -- (3, 0.05);
                    
                    \node[below] at (0, -0.1) {$0$};
                    \node[below] at (3, -0.1) {$1$};

                    \fill[red!30] (0, -0.05) rectangle (1.5, 0.05);
                    
                    \fill[blue!30] (1.5, -0.05) rectangle (3, 0.05);
                    
                    \draw[thick, black] (1.5, -0.05) -- (1.5, 0.05);
                    \node[below] at (1.5, -0.1) {$c(L)$};
                    
                    \node[red] at (0.5, 0.2) {$E[\text{Call}] \leq E[\text{Fold}]$};
                    \node[black] at (1.5, 0.2) {$E[\text{Call}] = E[\text{Fold}]$};
                    \node[blue] at (2.5, 0.2) {$E[\text{Call}] \geq E[\text{Fold}]$};
                    
                \end{tikzpicture}
                \caption{Caller's decision threshold $c(L)$. At hand strength of exactly $c(L)$, the caller is indifferent between calling and folding. Since the value of calling weakly increases with $y$, it must be weakly greater than 0 for $y > c(L)$ and weakly less than 0 for $y < c(L)$. Folding always has value 0.}
                \label{fig:caller_threshold}
            \end{figure}

            This proves that calling is weakly better than folding for all $y > c(s)$, and that folding is weakly better than calling for all $y < c(s)$, so the caller cannot improve their payoff by deviating from the strategy profile.  

        \item \textbf{Bettor's Deviation:}
            We need to consider a few cases.
            \begin{enumerate}
                \item $x < c(s)$:
                    These are hands and bet sizes for which the caller will call with only stronger hands (potential bluffs). The expected value of betting here is
                    \begin{align*}
                        \mathbb{E}[\text{bet } s | x] & = \mathbb{P}[\text{call with worse}] \cdot (1+s) - \mathbb{P}[\text{call with better}] \cdot s + \mathbb{P}[\text{fold}] \cdot 1 \\
                        & = 0 - (1-c(s)) \cdot (s) + c(s) \\
                        & = c(s) - (1-c(s)) \cdot s \\
                        & = x_2,
                    \end{align*}
                    with the last line coming from equation \ref{eq:bluffindiff}. The value of checking is always
                    \begin{align*}
                        \mathbb{E}[\text{check} | x] &= x.
                    \end{align*}
                    This means that (by design), the bettor is indifferent between checking and betting any amount at $x=x_2$. Importantly, the value of betting is independent of the hand strength $x$ while the value of checking is strictly increasing in $x$, so checking must be preferable for $x_2 < x < c(L)$ and betting must be preferable for $x < x_2$, which is exactly what our strategy profile does.
                    Because the value of bluffing is simply $x_2$ no matter the bet size or hand strength, we also know that the bettor cannot improve their payoff by bluffing with different bet sizes.
                \item $c(s) \leq x < x_3$:
                    These are hands and bet sizes for which the caller will at least sometimes call with weaker hands (potential value bets), but where the optimal strategy still checks. The expected value of betting here is
                    \begin{align*}
                        \mathbb{E}[\text{bet } s | x] & = \mathbb{P}[\text{call with worse}] \cdot (1+s) - \mathbb{P}[\text{call with better}] \cdot s + \mathbb{P}[\text{fold}] \cdot 1 \\
                        & = (x-c(s))(1+s) - (1-x) \cdot (s) + c(s) \\
                        & = s(2x - c(s) - 1) + x,
                    \end{align*}
                    while that of checking is 
                    \begin{align*}
                        \mathbb{E}[\text{check} | x] &= x.
                    \end{align*}
                    We know from \ref{eq:valueindiff} that
                    \begin{align*}
                        (x_3-c(L)) \cdot (1+L) - (1-x_3) \cdot (L) + c(L) & = x_3 \\
                        2x_3 - c(L) - 1 & = 0 \\
                    \end{align*}
                    Using our inequality $c(s) \leq x < x_3$ and the fact that $c(L)$ is the minimum of $c(s)$, we get
                    \begin{align*}
                        2x_3 - c(L) - 1 & = 0 \\
                        2x - c(s) - 1 & \leq 0.
                    \end{align*}
                    Substituting this into the expected value of betting, we get
                    \begin{align*}
                        \mathbb{E}[\text{bet } s | x] & = s(2x - c(s) - 1) + x \\
                        & \leq s \cdot 0 + x \\
                        & = x \\
                        &= \mathbb{E}[\text{check} | x].
                    \end{align*}
                    So no value can be gained by deviating from checking here.
                \item $x_3 \leq x < x_4$:
                    These are value bets where the bettor should bet the minimum. We need to show that the bettor cannot improve their payoff by either checking or by betting more.

                    The expected value of betting the minimum is
                    \begin{align*}
                        \mathbb{E}[\text{bet } L | x] & = L(2x - c(L) - 1) + x. 
                    \end{align*}
                    Again, we can use \ref{eq:valueindiff} to get 
                    \begin{align*}
                        2x_3 - c(L) - 1 & = 0 \\
                        2x - c(L) - 1 & \geq 0,
                    \end{align*}
                    since $x \geq x_3$. Substituting like before,
                    \begin{align*}
                        \mathbb{E}[\text{bet } L | x] & = L(2x - c(L) - 1) + x \\
                        & \geq L \cdot 0 + x \\
                        & = x \\
                        &= \mathbb{E}[\text{check} | x].
                    \end{align*}
                    So no value can be gained by deviating from betting to checking.

                    What about betting more? To show that this cannot improve the bettor's payoff, we show that the expected value of betting is weakly decreasing in $s$ for $x < x_4$, and must therefore be maximized at the lowest possible bet of $L$. 
                    
                    \begin{align*}
                        \frac{d}{ds} \mathbb{E}[\text{bet } s | x] & = -s c'(s) - c(s) + 2x - 1
                    \end{align*}
                    We know from \ref{eq:valueoptimality} that this equals $0$ when $x=v(s)$. We also know that $v(s)$ is at least $x_4$ for all $s \in [L, U]$, so $x < x_4 \leq v(s)$ for any such $s$. This means that
                    \begin{align*}
                        \frac{d}{ds} \mathbb{E}[\text{bet } s | x] & = -s c'(s) - c(s) + 2x - 1 \\
                        & \leq -s c'(s) - c(s) + 2x_4 - 1 \\
                        & = 0
                    \end{align*}
                    for any $s \in [L, U]$. Therefore, the expected value of betting is decreasing in $s$ for $x < x_4$, and must therefore be maximized at the lowest possible bet of $L$, so the bettor cannot improve their payoff by betting more.
                \item $x_4 \leq x < x_5$:
                    These are value bets where the bettor should bet an intermediate amount between $L$ and $U$. We need to show that the bettor cannot improve their payoff by either checking, betting less, or by betting more.

                    Rather than showing that checking is inferior to the optimal bet size, we show that checking is inferior to betting the minimum, which we will later show is inferior to the optimal bet size. Like the previous cases, the expected value of betting the minimum is
                    \begin{align*}
                        \mathbb{E}[\text{bet } L | x] & = L(2x - c(L) - 1) + x. 
                    \end{align*}
                    Like before, we know that $2x - c(L) - 1 \geq 0$ for $x \geq x_3$ (and in this case, $x \geq x_4 \geq x_3$). This means that
                    \begin{align*}
                        \mathbb{E}[\text{bet } L | x] & = L(2x - c(L) - 1) + x \\
                        & \geq L \cdot 0 + x \\
                        & = x \\
                        &= \mathbb{E}[\text{check} | x].
                    \end{align*}
                    So betting the minimum is at least as good as checking.

                    Now we show that betting any amount other than $v^{-1}(x)$ cannot gain value.
                    Let's again consider the derivative of the expected value with respect to $s$.
                    \begin{align*}
                        \frac{d}{ds} \mathbb{E}[\text{bet } s | x] & = -s c'(s) - c(s) + 2x - 1
                    \end{align*}
                    We know from \ref{eq:valueoptimality} that this is equal to $0$ when $x = v(s)$:
                    \begin{align*}
                        -s c'(s) - c(s) + 2v(s) - 1 &= 0.
                    \end{align*} 
                    This derivative is clearly an increasing function of $x$, so for $x < v(s)$, the expected value is decreasing in $s$. But since $v$ is an increasing function, $x < v(s)$ is equivalent to $v^{-1}(x) < s$ (essentially, our bet size is too large for the hand strength).
                    
                    This should make sense --- when our bet size is too large for the hand strength, the expected value of that bet is decreasing in the bet size, so smaller bets are more profitable. We can show the same for bets too small: when $x > v(s)$, the expected value is increasing in $s$. This is equivalent to saying that $v^{-1}(x) > s$, or our bet size is too small for the hand strength, so larger bets are more profitable. 
                    
                    We have shown that the expected value of betting is increasing for $s < v^{-1}(x)$, equal to $0$ at $s = v^{-1}(x)$, and decreasing for $s > v^{-1}(x)$, so the expected value of betting is maximized at $s = v^{-1}(x)$. This means that the bettor cannot improve their payoff by deviating from this bet size. In particular, they cannot benefit by betting the minimum, which in turn proves that they cannot benefit by checking, as we showed above.
                \item $x_5 \leq x \leq 1$:
                    These are value bets where the bettor should bet the maximum. We need to show that the bettor cannot improve their payoff by either checking or betting less.

                    The expected value of betting the maximum is
                    \begin{align*}
                        \mathbb{E}[\text{bet } U | x] & = U(2x - c(U) - 1) + x. 
                    \end{align*}
                    Plugging $s=U$, $x=v(U)=x_5$ into \ref{eq:valueoptimality}, we get
                    \begin{align*}
                        -U c'(U) - c(U) + 2x_5 - 1 &= 0.
                    \end{align*}
                    What happens when $x > x_5$? This expression must be greater than $0$, meaning the expected value of betting is increasing in $s$ for $x > x_5$. If it is always increasing in $s$ for such $x$, then it must be maximized at the largest possible bet of $U$, so the bettor cannot improve their payoff by betting less.

                    We can use the exact same logic as the previous case to show that checking cannot improve the payoff either.
                    
            \end{enumerate}
    \end{enumerate}
\end{customproof}

\section{Parameter Analysis: Complete Proofs}
\label{sec:parameter_analysis}

This appendix provides detailed proofs of how the upper limit $U$ and lower limit $L$ affect the Nash equilibrium strategies and expected payoffs in Limit Continuous Poker.

\subsection{Effect of Increasing $U$}

\subsubsection{Expected Payoff of Value-Betting Hands}

It may seem unsurprising that strong hands become less likely to get called as limits increase, but what about the actual expected payoff of these hands? Does the expected value of a specific hand strength in Nash equilibrium continue increasing as we increase $U$? The answer is no. In fact, for any fixed hand strength $x$, the expected payoff of that hand increases in $U$ only up to a certain threshold, after which it decreases (see Figure \ref{fig:ev_x_vs_U}). This feels counterintuitive; increasing $U$ only gives the bettor more options, so how is it possible that the expected payoff of individual hands decreases? And which hands are gaining expected payoff to offset this? This is a surprising result, and it is worth exploring in more detail.

\begin{theorem}
    \label{thm:payoff_increasing}
    For any value-betting hand strength $x$ and any $L, U$, $\frac{d}{dU} EV(x) < 0$ if

    $$x < \max\left(v(U), \frac{1}{2(1+U)} \left( U \frac{\partial x_2}{\partial U} + \frac{U^2 + x_2(1 + 2U)}{1+U} \right) \right),$$

    and $\frac{d}{dU} EV(x) > 0$ otherwise.
\end{theorem}

To parse this in English: if we fix the bettor's hand strength $x$, the expected payoff to the bettor in our constructed Nash equilibrium is decreasing in $U$ if $x$ falls below a certain threshold, but increasing in $U$ if $x$ is above this threshold. Specifically, this threshold is greater than the hand strength $v(U)$ which bets the maximum.

Before proving the theorem, we will walk through some lemmas which explore how all the relevant variables change as we increase $U$, including the bluffing threshold $x_2$, the bet size $v^{-1}(x)$, and the calling cutoff $c(s)$.

\subsubsection{Bluffing Threshold}

We begin by showing that $x_2$, the boundary hand strength between bluffing and checking, is increasing in $U$. This means that for fixed $L$, increasing the upper limit $U$ makes the bettor bluff with more hands.

\begin{lemma}
    \label{lem:x2_increasing}
    For any $L, U$,
    $$ \frac{d x_2}{d U} > 0. $$
\end{lemma}

\begin{customproof}
    Recall that $x_2$ is given by:

    \[ x_2 = \frac{r^{3} + t^{3} - 1}{r^{3} + t^{3} - 7} \]

    where $r = L/(1+L)$ and $t = 1/(1+U)$. We can use the chain rule to differentiate $x_2$ with respect to $U$:

    \[ \frac{d x_2}{d U} = \frac{\partial x_2}{\partial t} \frac{d t}{d U}. \]

    Note that $r$ has no dependence on $U$. We compute:

    \[ \frac{\partial x_2}{\partial t} = \frac{-18 t^{2}}{\left(r^{3} + t^{3} - 7\right)^{2}}, \quad \frac{d t}{d U} = - \frac{1}{(1+U)^2}. \]

    Therefore,

    \[ \frac{d x_2}{d U} = \frac{-18 t^{2}}{\left(r^{3} + t^{3} - 7\right)^{2}} \cdot \left(- \frac{1}{(1+U)^2}\right) = \frac{18 t^{2}}{(1+U)^2\left(r^{3} + t^{3} - 7\right)^{2}} > 0, \]

    which is positive since $r, t \in (0, 1)$ implies $r^3 + t^3 - 7 < 0$.

\end{customproof}

\subsubsection{Bet Size}

We now show that if we fix $x$ at any intermediate value-betting hand strength (betting neither the minimum nor maximum bet size) and then increase $U$, the bet size $s$ made by $x$ decreases. The intermediate value-betting hands are exactly $x \in [x_3, v(U)]$ and their bet sizes are given by $s = v^{-1}(x)$, so we get the following lemma:

\begin{lemma}
    \label{lem:v_inverse_decreasing}
    For any fixed $x \in [x_3, v(U)]$,
    \[
        \frac{d}{dU} v^{-1}(x) < 0
    \]
\end{lemma}

\begin{customproof}
    Recall that
    $$v^{-1}(x) = -\frac{\sqrt{(4 x-4) (2 x_2-2)}}{4 x-4}-1,$$
    where $x_2 = \frac{r^3 + t^3 - 1}{r^3 + t^3 - 7}$ with $t = 1/(1+U)$. Importantly, $v^{-1}(x)$ is only dependent on $U$ through $x_2$, which in turn depends on $U$ only through $t$. Using the chain rule:
    \begin{align*}
        \frac{d}{dU} v^{-1}(x) & = \frac{\partial v^{-1}(x)}{\partial x_2} \frac{\partial x_2}{\partial t} \frac{dt}{dU}
    \end{align*}

    We compute each factor:

\begin{align*}
    \frac{\partial v^{-1}(x)}{\partial x_2} & = - \frac{1}{\sqrt{(4 x-4) (2 x_2-2)}} = - \frac{1}{(v^{-1}(x)+1)(4-4x)} < 0
\end{align*}

    which is negative since $x \in [0, 1]$ and $v^{-1}(x) >0$. From Lemma \ref{lem:x2_increasing}, we know that $\frac{\partial x_2}{\partial t} \frac{dt}{dU} > 0$.

    Therefore, the product of the three terms is always negative, so the bet size of intermediate bets is decreasing in $U$.
\end{customproof}

\subsubsection{Calling Cutoff}

Recall that $c(s)$ is defined as the minimum hand strength $y$ which should call a bet of size $s$ and is given in Nash equilibrium by:

$$c(s) = \frac{x_2 + s}{s+1}$$

We are specifically interested in how $c(v^{-1}(x))$ varies with $U$ for $x \in [x_3, v(U)]$, since this represents how the calling cutoff changes both directly from a strategic change, as well as indirectly due to the lower bet size $s$. It turns out that the calling cutoff is increasing in $U$ for all $x \in [x_3, v(U)]$. This is surprising because we just showed that the bet size $s$ is decreasing in $U$, and we expect smaller bets to be called more often. For reasons we will see later, this effect is overpowered by a strategic shift for the caller, who calls less often for all bet sizes as $U$ increases.

\begin{lemma}
    \label{lem:c_increasing}
    For any fixed $x \in [x_3, v(U)]$,
    \[
        \frac{d}{dU} c(v^{-1}(x)) > 0
    \]
\end{lemma}

\begin{customproof}
    As mentioned above, $c(s) = \frac{x_2 + s}{s+1}$ is dependent on $U$ in two distinct ways: directly through $x_2 = x_2(t)$ where $t = 1/(1+U)$, and indirectly through the bet size $s = v^{-1}(x)$. We use the multivariate chain rule:
    \begin{align*}
        \frac{d}{dU} c(v^{-1}(x)) & = \frac{\partial c(s)}{\partial s} \frac{d v^{-1}(x)}{d U} + \frac{\partial c(s)}{\partial x_2} \frac{\partial x_2}{\partial t} \frac{dt}{dU}
    \end{align*}

    The partial derivatives of $c(s)$ are:

    $$ \frac{\partial c(s)}{\partial s} = \frac{1-x_2}{(s+1)^2} \; \; \; \text{and} \; \; \; \frac{\partial c(s)}{\partial x_2} = \frac{1}{s+1}. $$

    Substituting $s = v^{-1}(x)$, we can simplify the first term using the fact that $(v^{-1}(x)+1)^2 = (1-x_2)/(2-2x)$:

    \begin{align*}
        \frac{\partial c(s)}{\partial s} \bigg|_{s=v^{-1}(x)} & = \frac{1-x_2}{(v^{-1}(x)+1)^2} = 2-2x
    \end{align*}

    From Lemma \ref{lem:v_inverse_decreasing}, we have:

    $$ \frac{d v^{-1}(x)}{d U} = \frac{\partial v^{-1}(x)}{\partial x_2} \frac{\partial x_2}{\partial t} \frac{dt}{dU} = \frac{-1}{(v^{-1}(x)+1)(4-4x)} \frac{\partial x_2}{\partial t} \frac{dt}{dU}. $$

    Substituting everything:

    \begin{align*}
        \frac{d}{dU} c(v^{-1}(x)) & =
        (2-2x) \cdot \frac{-1}{(v^{-1}(x)+1)(4-4x)} \cdot \frac{\partial x_2}{\partial t} \frac{dt}{dU}
        + \frac{1}{v^{-1}(x)+1} \cdot \frac{\partial x_2}{\partial t} \frac{dt}{dU}\\
        & = \frac{1}{v^{-1}(x)+1} \cdot \frac{\partial x_2}{\partial t} \frac{dt}{dU} \cdot \left( -\frac{2-2x}{4-4x} + 1\right) \\
        & = \frac{1}{v^{-1}(x)+1} \cdot \frac{\partial x_2}{\partial t} \frac{dt}{dU} \cdot \frac{1}{2}
    \end{align*}

    All three factors are positive: $v^{-1}(x) > 0$, and by Lemma \ref{lem:x2_increasing}, $\frac{\partial x_2}{\partial t} \frac{dt}{dU} > 0$. Therefore, the calling cutoff is increasing in $U$ for all $x \in [x_3, v(U)]$.

\end{customproof}

\subsubsection{Proof of Theorem \ref{thm:payoff_increasing}}

Having these tools, we can now finally return to the proof of Theorem \ref{thm:payoff_increasing}.

\begin{customproof}
    Recall the expected payoff of a value-betting hand $x$:

    \begin{align*}
        EV(x) & = \frac{1}{2} c(s) + (x - c(s)) \left(s+\frac{1}{2}\right) + (1-x) \left(-s-\frac{1}{2}\right)
    \end{align*}

    We break the proof into two cases:

    \textbf{Case 1 ($x > v(U)$):} In this case, hand $x$ bets the maximum amount $U$. Recall that $c(s)$ is implicitly a function of $s$ and $x_2$, which is itself a function of $U$. Using the multivariate chain rule, the derivative at $s=U$ is:

    \begin{align*}
        \frac{d}{dU} EV(x) & = \frac{\partial EV(x)}{\partial s} \bigg|_{s=U} + \frac{\partial EV(x)}{\partial c(s)} \left( \frac{\partial c(s)}{\partial s} + \frac{\partial c(s)}{\partial x_2} \frac{\partial x_2}{\partial U} \right) \bigg|_{s=U} \\
        & = \left( \frac{\partial EV(x)}{\partial s} + \frac{\partial EV(x)}{\partial c(s)} \frac{\partial c(s)}{\partial s}\right)\bigg|_{s=U}   +  \left( \frac{\partial EV(x)}{\partial c(s)} \frac{\partial c(s)}{\partial x_2} \frac{\partial x_2}{\partial U} \right) \bigg|_{s=U}
    \end{align*}

    We want to know exactly when the above expression is positive. The partial derivatives to plug in are:

    \begin{align*}
        \frac{\partial EV(x)}{\partial s} & = 2x - 1 - c(s) \\
        \frac{\partial EV(x)}{\partial c(s)} & = - s \\
        \frac{\partial c(s)}{\partial s} & = \frac{1-x_2}{(s+1)^2} \\
        \frac{\partial c(s)}{\partial x_2} & = \frac{1}{s+1} \\
    \end{align*}

    We can leave $\frac{\partial x_2}{\partial U}$ as a free variable for now, since it is always positive by Lemma \ref{lem:x2_increasing}, and is independent of $x$. Plugging these in and rearranging terms, we can say that $EV(x)$ is increasing in $U$ if

    \begin{align*}
        x & > \frac{1}{2(1+U)} \left( U \frac{\partial x_2}{\partial U} + \frac{U^2 + x_2(1 + 2U)}{1+U} \right)
    \end{align*}

    Where nothing on the right hand side is dependent on $x$. This means that for any fixed $L, U$, this gives a threshold value for $x$ below which $EV(x)$ is decreasing in $U$, and above which it is increasing in $U$.

    \textbf{Case 2 ($x_3 < x < v(U)$):} In this case, hand $x$ makes an intermediate-sized bet $s = v^{-1}(x)$. There are two distinct factors influencing the derivative $\frac{d}{dU} EV(x)$, namely the change in bet size $s = v^{-1}(x)$ and the change in calling cutoff $c(v^{-1}(x))$. By the multivariate chain rule, we can express the derivative as:

    \begin{align*}
        \frac{d}{dU} EV(x) & = \frac{\partial EV(x)}{\partial s} \frac{d v^{-1}(x)}{d U} + \frac{\partial EV(x)}{\partial c(s)} \frac{d c(v^{-1}(x))}{\partial U}
    \end{align*}

    The partial derivatives of $EV(x)$ are:

    \begin{align*}
        \frac{\partial EV(x)}{\partial s} & = 2x - 1 - c(s) \\
        \frac{\partial EV(x)}{\partial c(s)} & = - s
    \end{align*}

    The second is clearly negative. We can verify that the first must be positive if we go back to the constraints which gave us the Nash equilibrium. For the bet size $v^{-1}(x)$ to be optimal, we required that

    $$ -s \frac{\partial c(s)}{\partial s} - c(s) + 2v(s) - 1 = 0,$$

    or equivalently, if we substitute $s= v^{-1}(x)$ and $v(s) = x$ and rearrange:

    $$ 2x -1 -c(v^{-1}(x)) = v^{-1}(x) \frac{\partial c(s)}{\partial s} > 0,$$

    since $\frac{\partial c(s)}{\partial s} = \frac{1-x_2}{(s+1)^2} > 0$, and $s$ is positive by definition.

    We know from Lemma \ref{lem:v_inverse_decreasing} that $\frac{d v^{-1}(x)}{d U} < 0$ and from Lemma \ref{lem:c_increasing} that $\frac{d c(v^{-1}(x))}{\partial U} > 0$.

    Combining everything, we see that both terms in $\frac{d}{dU} EV(x)$ are products of negative and positive, making both terms negative. Therefore, the expected payoff is decreasing in $U$ for all $x \in [x_3, v(U)]$.

\end{customproof}

\section{Payoff Analysis: Complete Proofs}
\label{sec:payoff_analysis}

This appendix provides detailed proofs regarding the expected payoffs and values of different hand strengths in the constructed Nash equilibrium.

\subsection{Expected Value of Bettor's Hand Strengths}

In addition to considering payoffs for a specific bettor-caller hand combination, we can also consider the expected value of a given hand for the bettor, not knowing the hand of the caller.

\begin{theorem}
    \label{thm:ev_bettor}
    Let $EV(x)$ denote the expected value of a hand $x$ in our constructed Nash equilibrium:
    \begin{equation}
        EV(x) = \begin{cases}
            x_2-\frac{1}{2} & \text{if } x \leq x_2 \\
            x-\frac{1}{2} & \text{if } x_2 < x \le x_3 \\
            x(2L + 1) - L(c(L) + 1) - \frac{1}{2} & \text{if } x_3 < x < v(L) \\
            x(2v^{-1}(x) + 1) - v^{-1}(x)(c(v^{-1}(x)) + 1) - \frac{1}{2} & \text{if } v(L) \leq x \leq v(U) \\
            x(2U + 1) - U(c(U) + 1) - \frac{1}{2} & \text{if } x > v(U).
        \end{cases}
    \end{equation}
\end{theorem}

\begin{customproof}
    For bluffing hands $x \leq x_2$, we can use a simple argument to show that the hand strength $x$ is actually irrelevant. The bettor never gets called by worse hands, so either the caller folds or calls with the best hand. In either case, the bettor's payoff has no dependence on $x$, so the payoff must be the same for all $x \leq x_2$ (otherwise, the lower-payoff hands would imitate the strategy of the higher-payoff hands). We know that at $x=x_2$, the bettor is indifferent between bluffing and checking, so the payoff must be $x_2-\frac{1}{2}$ for all $x \leq x_2$.

    For checking hands $x_2 \leq x \leq x_3$, the bettor wins only the ante exactly when they have the best hand, which happens with probability $x$. The value of the ante is $\frac{1}{2}$, so the bettor's expected value is $x-\frac{1}{2}$.

    For any value betting hand, the bettor has three cases to consider: the caller folds, the caller calls with a worse hand, or the caller calls with the best hand. We simply sum the expected value of each of these cases.

    \begin{align*}
        EV(x) & = \frac{1}{2} c(s) + (x - c(s)) \left(s+\frac{1}{2}\right) + (1-x) \left(-s-\frac{1}{2}\right)
    \end{align*}

    The last three cases come from substituting $L, v^{-1}(x), U$ for $s$ in the expression above and simplifying.
\end{customproof}

\subsection{Monotonicity of Expected Value}

We can quickly verify that the bettor's expected value is increasing in $x$. This must be the case, since with any given hand strength, the bettor can always choose to imitate the Nash equilibrium strategy of a weaker hand, so the stronger hand must be at least as good in expectation.

\begin{theorem}
    \label{thm:ev_increasing}
    For any fixed $L, U$, the bettor's expected value $EV(x)$ is increasing in $x$.
\end{theorem}

\begin{customproof}
    It is clear from inspection that any checking EV is higher than that of a bluff and that the checking EV is increasing in $x$. We also know that at $x=x_3$, the bettor is indifferent between checking and betting, so the EV of checking and betting must be equal at this point. Therefore, we only need to show that within the checking and value betting regions, the EV is increasing in $x$. This is obvious for checking hands. For value betting hands, we can take the derivative of the expression for $EV(x)$ with respect to $x$ and show that it is always positive. We consider the max and min betting hands first:

    \begin{align*}
        s = U \implies \frac{d}{dx} EV(x) & = 2U + 1 > 0  \\
        s = L \implies \frac{d}{dx} EV(x) & = 2L + 1 > 0
    \end{align*}

    For intermediate value betting hands, we can use the chain rule to show that the EV is increasing in $x$:

    \begin{align*}
        \frac{d}{dx} EV(x) & = \frac{\partial EV(x)}{\partial x} + \frac{dEV(x)}{ds} \frac{d s}{d x} \\
    \end{align*}

    From the value optimality condition (equation \ref{eq:valueoptimality} in the main text), we know that $\frac{dEV(x)}{ds} = 0$ (otherwise, the bettor could gain value by varying the bet size). Clearly, $\frac{\partial EV(x)}{\partial x} > 0$.  Therefore, the derivative is positive, and the EV is increasing in $x$ for all value betting hands.
\end{customproof}

\subsection{Discussion: Strong Hands and Risk-Reward Tradeoffs}

It is worth noting that the bettor's strongest hands (right edge) actually seem to become less likely to make any profit more than the ante as limits increase. These strongest hands make very large bets, which force all but the strongest hands to fold, but win huge pots when they do get called.

In more complicated poker variants, it is common to ``slowplay" strong hands by checking or making small bets to induce bluffs from the opponent. In LCP, there is only one betting round and the caller is not allowed to raise, both of which make slowplaying obsolete. With extremely strong hands, the benefit of winning a large pot when betting big outweighs the lower likelihood of getting called.

This strategic pattern demonstrates a fundamental tension in poker: extracting maximum value from strong hands requires finding the optimal balance between bet size (which determines pot size when called) and calling frequency (which decreases as bet size increases). In our constructed Nash equilibrium, the strongest hands resolve this tension by accepting a lower calling frequency in exchange for winning much larger pots.

\pagebreak

\end{document}